\documentclass[journal]{IEEEtran}
\usepackage{cite}
\usepackage{algpseudocode}
\usepackage[pdftex]{graphicx}
\usepackage{amsmath}
\usepackage{empheq}
\usepackage{array}

\usepackage{caption}
\usepackage{subcaption}
\usepackage[font=scriptsize]{caption}
\usepackage{url}
\usepackage{amsmath,amsthm,amssymb}
\usepackage{txfonts}
\newtheorem{theorem}{Theorem}

\usepackage{tabularx,booktabs}
\newcolumntype{C}{>{\centering\arraybackslash}X} 
\usepackage{color}
\usepackage{verbatim}

\usepackage{algpseudocode}
\usepackage{algorithm}

\begin{document}

\title{ Resilient Autonomous Control of Distributed Multi-agent Systems in Contested Environments}

\author {Rohollah Moghadam, ~\IEEEmembership{Student Member,~IEEE}, {Hamidreza Modares,~\IEEEmembership{Member,~IEEE}} \\
\thanks{ Rohollah Moghadam and Hamidreza Modares are with the Department
of Electrical and Computer Engineering, Missouri University of Science and Technology,
Emerson Hall, 301 W 16th St, Rolla, MO 65409 USA (e-mails: moghadamr@mst.edu, modaresh@mst.edu).}
\thanks{Manuscript received December 30, 2017;}}


\maketitle

\begin{abstract}
An autonomous and resilient controller is proposed for leader-follower multi-agent systems under uncertainties and cyber-physical attacks. The leader is assumed non-autonomous with a nonzero control input, which allows changing the team behavior or mission in response to environmental changes. A resilient learning-based control protocol is presented to find optimal solutions to the synchronization problem in the presence of attacks and system dynamic uncertainties. An observer-based distributed H${_\infty}$ controller is first designed to prevent propagating the effects of attacks on sensors and actuators throughout the network, as well as to attenuate the effect of these attacks on the compromised agent itself. Non-homogeneous game algebraic Riccati equations are derived to solve the H$\infty$ optimal synchronization problem and off-policy reinforcement learning is utilized to learn their solution without requiring any knowledge of the agent's dynamics. A trust-confidence based distributed control protocol is then proposed to mitigate attacks that hijack the entire node and attacks on communication links. A confidence value is defined for each agent based solely on its local evidence. The proposed resilient reinforcement learning algorithm employs the confidence value of each agent to indicate the trustworthiness of its own information and broadcast it to its neighbors to put weights on the data they receive from it during and after learning. If the confidence value of an agent is low, it employs a trust mechanism to identify compromised agents and remove the data it receives from them from the learning process. Simulation results are provided to show the effectiveness of the proposed approach.   

\end{abstract}


\section{Introduction}

The term autonomous is a combination of the Greek words auto (self) and nomous (law, rule) \cite{Antsaklis1991}. In system theory, an autonomous system is a system that is self-governing and does not explicitly depend on the independent variable. If the independent variable is time, these systems are also called time-invariant systems. On the other hand, in control theory, an autonomous control is a self-governing control system in the sense that it acts independently and does not rely on prior knowledge of the system dynamics and human intervention. An autonomous controller should be able to learn from what it perceives to compensate for partial or incorrect prior knowledge.

Autonomous control design for a multi-agent system (MAS) has gained significant interest due to applications in a variety of disciplines from robot swarms to power systems and wireless sensor networks. In a distributed MAS, decisions are made locally using only agents' available information. This provides scalability, flexibility and avoids a single point of failure \cite{Olfati2007,Jadbabaie2003,Fax2004,Ren2007}. However, designing autonomous controllers for MASs is challenging and requires learning from experience. Moreover, distributed MASs are prone to cyber-physical attacks. Due to their networked nature, attacks can escalate into disastrous consequences and significantly degrade the performance of the entire network \cite{Rohollah2017}. In a contested environment with adversarial inputs, corrupted data communicated by a single compromised agent can be propagated to the entire network through its neighbors. This corrupted data will be used by autonomous agents for learning which misleads the entire network and, consequently, causes no emergent behavior or an emergent misbehavior. The main bottleneck in deploying successful distributed MASs is designing secure control protocols that can learn about system uncertainties while showing some level of functionality in the presence of cyber-physical attacks.  

Reinforcement learning (RL) \cite{sutton1998reinforcement,powell2007approximate,lewis2012reinforcement}, inspired by learning mechanisms observed in mammals, has been successfully used to learn optimal solutions online in single agent systems for both regulation and tracking control problems \cite{Liu2014Decentralized,JIANG2012Computational,Luo2015Off,Dixon2017concurrent,Ni2013Adaptive,BahareTNNLS,vamvoudakis2012online} and recently for MASs \cite{ABOUHEAF20143038,Babu2008,Zuo2017LewisRL}. Existing RL-based controllers for leader-follower MASs assume that the leader is passive and without any control input. In this case, the leader is not able to react to environmental or mission changes by replanning its trajectories. On the other hand, existing active leader controllers (e.g., \cite{HONG20061177,Yang2018,Li2013TAC}) are not autonomous as they require having complete knowledge of the leader and agent's dynamics. Moreover, these approaches are generally far from optimal and only take into account the stability, which is the bare minimum requirement. Finally, existing learning-based RL solutions to MASs are not resilient against cyber-physical attacks.

Resilient control protocols for MASs have been designed  in the literature \cite{Pasqualetti2012consensus,LeBlanc2013IJSAC,Zhang2012AccRobustness,Teixeira2010,Amin2013TCST,Weerakkody2017,Feng2017ACC,Khazraei2017ACC,Sundaram2011,Zeng2014Resilient,LeBlanc2013resilient,Abhinav2017,RohollahIET} to mitigate attacks. Most of the existing approaches either use the discrepancy among the state of agents and their neighbors to detect and mitigate attacks, or use an exact model of agents to predict expected normal behavior and, thus, detect an abnormality caused by attacks. However, we will show that a stealthy attack on one agent can cause an emergent misbehavior in the network with no discrepancy between agent's states and, therefore, the former approaches cannot mitigate these type of attacks. Moreover, this discrepancy could be a result of a legitimate change in the state of the leader. Blindly rejecting the agent's neighbor information can harm the network connectivity and the convergence of the network. On the other hand, model-based approaches require a complete knowledge of the agent's dynamics, which may not be available in many practical applications and avoid the design of autonomous controllers. H$_\infty$ control protocols have also been proposed to attenuate the effect of disturbances in MASs \cite{Cao2010Optimal,Li2017ResilientAutonomous,JIAO2016Automatica}. However, as shown in this paper, standard H$_\infty$ control protocols can be misled and become entirely ineffective by a stealthy attack. To the author's knowledge, designing an autonomous and resilient controller that does not require any knowledge of the agent's dynamic and can survive against cyber-physical attacks has not been investigated yet.

This paper presents an autonomous and resilient distributed control protocol for leader-follower MASs with a non-autonomous leader. To alleviate the effects of attacks on the MASs, a distributed observer-based control protocol is first developed to prevent corrupted sensory data caused by attacks on sensors and actuators from propagating across the network. To this end, only the leader communicates its actual sensory information and other agents estimate the leaders' state using a distributed observer. To further improve resiliency, distributed H$_\infty$ control protocols are designed to attenuate the effect of the attacks on the compromised agent itself. Non-homogeneous game algebraic Riccati equations (ARE) are derived for solving the optimal H$_\infty$ synchronization problem for each agent. An off-policy RL algorithm is developed to learn the solutions of the non-homogeneous game ARE without requiring the complete knowledge of the agent's dynamics. To avoid the usage of corrupted data coming from compromised neighbors during and after learning, a trust-confidence based control protocol is developed for attacks on communication links and attacks that hijack the entire node. A confidence value is defined for each agent based solely on its local evidence. Then, each agent communicates its confidence value with neighbors to indicate the trustworthiness of its own information. Moreover, a trust value is defined for each neighbor to determine the significance of the incoming information. The agent incorporates these trust values along with the confidence values received from neighbors in its update law to eventually isolate the compromised agent.

\section{Preliminary}
In this section, a background of the graph theory is provided. A directed graph $\mathcal{G}$ consists of a pair $\left( {\mathcal{V},\mathcal{E}} \right)$ in which $\mathcal{V}{\text{ =  }}\{ {v_1}, \cdots ,{v_N}\}$ is a set of nodes and $\mathcal{E} \subseteq \mathcal{V} \times \mathcal{V}$ is a set of edges. The adjacency matrix is defined as $\mathcal{A} = \left[ {{a_{ij}}} \right]$, with ${a_{ij}} > 0$  if $({v_j},{v_i}) \in \mathcal{E}$, and ${a_{ij}} = 0$  otherwise. The set of nodes ${v_i}$ with edges incoming to node ${v_j}$ is called the neighbors of node ${v_i}$, namely ${\mathcal{N}_i} = \{ {v_j}:({v_j},{v_i}) \in \mathcal{E}\}$. The graph Laplacian matrix is defined as $\mathcal{L} = D - \mathcal{A}$, where $D = diag({d_i})$ is the in-degree matrix, with ${d_i} = \sum\nolimits_{j \in {N_i}} {{a_{ij}}} $ as the weighted in-degree of node ${v_i}$. A (directed) tree is a connected digraph that in-degree of every node is one, except the root node. A directed graph has a spanning tree if there exists a directed tree that connects all nodes of the graph. A leader can be pinned to multiple nodes, resulting in a diagonal pinning matrix $G = {diag}\left( {{b_i}} \right) \in {\mathbb{R}^{N \times N}}$ with the pinning gain ${b_i} > 0$ when the node has access to the leader node and ${b_i} = 0$, otherwise. ${{\mathbf{1}}_N}$ is the ${N}$-vector of ones and $\operatorname{Im} (R)$ denotes the range space of $R$. $\lambda_{min}(A)$ denotes the minimum eigenvalue of matrix A. $\left\| . \right\|$ denotes the Euclidean norm for vectors or the induced 2-norm for matrices. The notation $A \otimes B$ is Kronecker product of matrices $A$ and $B$.
\smallskip

\noindent
\textbf{Assumption 1.} The communication graph has a spanning tree, and the leader is pinned to at least one root node.

\section{Standard Synchronization Control Protocols and Their Vulnerability to attacks}
In this section, the standard synchronization control protocol for MASs is reviewed and its vulnerability to attacks is examined. 
Consider $N$ agents with identical dynamics given by
\begin{equation}\label{eq3}
{{\dot x}_i} = A{x_i} + B{u_i} + D{\omega _i}
\end{equation}
where ${x_i}(t) \in {\mathbb{R}^n}$ and ${u_i}(t) \in {\mathbb{R}^m}$ are the state and control input of agent $i$, respectively. ${\omega _i}(t)  \in {\mathbb{R}^d}$ denotes the attack signal injected into agent $i$. $A$, $B$, and $D$ are the drift, input, and attack dynamics, respectively.
\smallskip

\noindent
\textbf{Assumption 2.} The pair $(A,B)$ is stabilizable. 

Let the leader dynamics be non-autonomous, i.e., the control input of the leader is a nonzero signal, and is given by 
\begin{equation}\label{eq4}
{{\dot \zeta }_0} = A{\zeta _0}+Bv_0
\end{equation}
where ${\zeta _0}(t) \in {\mathbb{R}^n}$ and ${v_0}(t) \in {\mathbb{R}^m}$ denote the state and input of the leader, respectively. $A$ and $B$ are the same as other agents. 
\smallskip

\noindent
\textbf{Assumption 3.} The control input $v_0$ is given and bounded, i.e., there exists a positive constant $v_m$ such that $\left\| v_0\right\| \leq v_m$.

Define the tracking error for agent $i$ as
\begin{equation}\label{eq5}
{\varepsilon_i} = {x_i} - {\zeta_0}
\end{equation}
 
Define the local neighborhood tracking error ${e_i} \in {\mathbb{R}^n}$ for agent $i$ as  \cite{Zhang2011Optimal}
\begin{equation}\label{eq6}
{e_i} = \sum\limits_{j \in {N_i}} {{a_{ij}}({x_j} - {x_i}) + {b_i}({\zeta _0}-{x_i})} 
\end{equation}
where ${b_i} \geqslant 0$ is the pinning gain, and ${b_i} > 0$ for at least one root node $i$. The standard distributed tracking control protocol is then given by \cite{li2014cooperative}
\begin{equation}\label{Cnt_S}
{u_i} = c_1K{e_i}+c_2h(K{e_i})
\end{equation}
where $c_1$ and $c_2$ are positive scalar coupling gains, and $K$ is a design matrix gain. $h(.)$ is a nonlinear function defined for $x\in {\mathbb{R}}^n$ such that 
\begin{equation}\label{g_Func}
h(x) = 
\begin{cases}
    \frac{x}{\left\| x \right\|},& if \,\,{\left\| x \right\| \neq 0}\\
    0,          & if \,\, {\left\| x \right\| = 0}
\end{cases}
\end{equation}

It can be seen that the following condition is required to assure synchronization
\begin{equation}\label{eq8}
{\varepsilon_i} \to 0 \Leftrightarrow {x_i} \to {\zeta_0}
\end{equation}

\begin{theorem}\cite{li2014cooperative} Consider the agent dynamics \eqref{eq3}-\eqref{eq4} with  ${\omega _i}(t)=0$. Suppose Assumptions 1 and 3 hold. Then, agents synchronize to the leader, i.e., ${\varepsilon_i} \to 0 \,\, \forall i\in N$ under the controller \eqref{Cnt_S} with $c_1 \geq (1/\lambda_{min}(\mathcal{L}+G))$, $c_2 \geq v_m$ and $K=-B^TP^{-1}$, where $\lambda_{min}$ is the minimum eigenvalue of $(\mathcal{L}+G)$ and $P>0$ is a solution to the linear matrix inequality $AP+PA^T-2BB^T<0$. 
\end{theorem}

\noindent
\textbf{Definition 1.} In a graph, agent $i$ is \textbf{reachable} from agent $j$ if there exists a directed path of any length from node $j$ to node $i$.
\vspace{3pt}
 
\noindent
\textbf{Definition 2:} An agent is called a \textbf{disrupted/compromised} agent, if it is directly under attack. Otherwise, it is called an \textbf{intact agent}.
\smallskip

It is shown in \cite{li2014cooperative} for the proof of Theorem 1 that in the absence of attack, if the controller is designed to make the local neighborhood tracking error for all agents go to zero, the synchronization is guaranteed. In the following theorem, however, it is shown that even though the local neighborhood tracking \eqref{eq6} goes to zero for all agents, in the presence of a specific designed attack, it does not guarantee synchronization for intact agents that have a path to a compromised agent. Note that the leader is assumed to be a trusted agent with more advanced sensors and with higher security. Note also that the leader does not receive any information from other agents, which makes it secure against attacks on other agents and the communication network.
\vspace{3pt}

\noindent
\textbf{Lemma 1 \cite{lewis2013cooperative}.} Let $\Sigma$ be a diagonal matrix with at least one nonzero positive element, and $L$ be the Laplacian matrix. Then, $(L+\Sigma)$ is a nonsingular M-matrix. 
\begin{theorem}
Consider the MAS \eqref{eq3}-\eqref{eq4} with the control protocol \eqref{Cnt_S}. Assume that agent $i$ is under an attack $\omega_i$ that is generated by $\dot \omega_i = \Gamma \omega_i\,$, where the eigenvalues of $\,\Gamma$ are a subset of the eigenvalues of the agent's dynamic $A$. Then, intact agents that are reachable from agent $i$ do not synchronize to the leader, while their local neighborhood tracking error \eqref{eq6} is zero.
\end{theorem}
\begin{proof}
See Appendix A.
\end{proof}

\noindent
\textbf{Remark 1.} H$_\infty$ is one of the most common and effective approaches to attenuate disturbances. However, Theorem 2 implies that the standard H$_\infty$ controllers for MASs that use the exchange of relative information can be bypassed by the attacker. This is because although the goal of the H$_\infty$  is to attenuate the effect of adversarial input on the local neighborhood tracking error, Theorem 2 shows that the attacker can make the local neighborhood tracking error go to zero, while agents are far from synchronization. Therefore, a different controller framework and H$_\infty$ controller is presented in this paper that guarantees attenuating attacks on sensors and actuators of a compromised agent.

\section{The Proposed Attack Mitigation Approach}
In this section, the proposed resilient control approach is presented. First, a distributed observer-based H$_\infty$ control protocol is developed to not only prevent attacks on physical components, i.e., attacks on sensors and actuators (we call them \textbf{Type 1} attacks), from propagating throughout the network but also attenuate their effect on the compromised agent. Then, a trust-confidence based control protocol is examined to identify and isolate neighbors that are compromised by attacks on the communication network or attacks that take over the control of a compromised agent (we call them \textbf{Type 2} attacks). Figure \ref{fig:Fig4} shows the structure of the proposed control framework. 
\begin{figure}[!ht]
\begin{center}
\includegraphics[width=1\linewidth,height=2in]{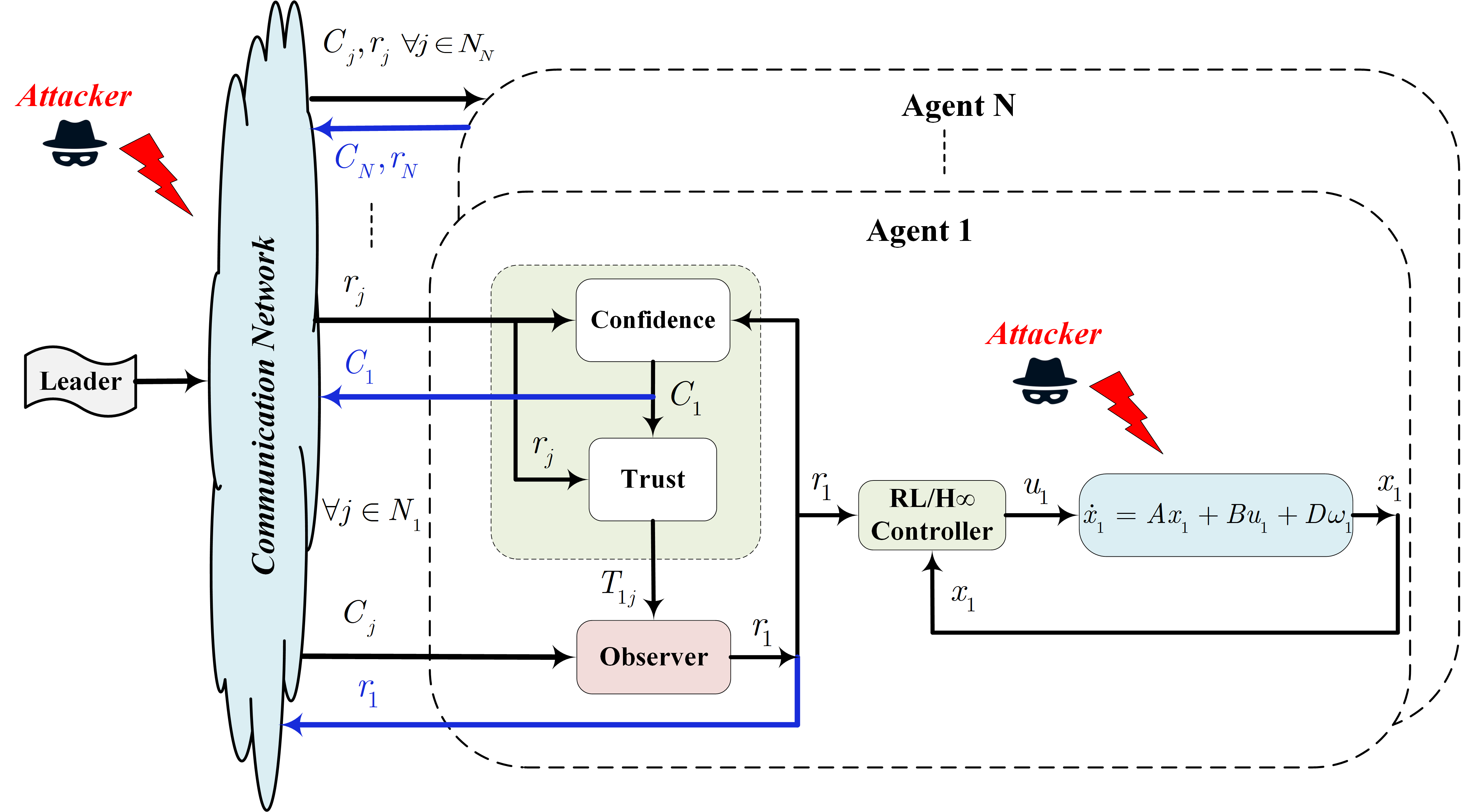}
\vspace{-5pt}
\caption{The proposed autonomous and secure distributed H${_\infty}$ controller for multi-agent systems. The observer and H${_\infty}$ controller allow mitigation of Type 1 attacks and designing autonomous controllers for agents. The trust-confidence mechanism mitigates Type 2 attacks. }\label{fig:Fig4}
\captionsetup{justification=centering}
\vspace{-20pt}
\end{center}
\end{figure}
\subsection{Overall structure of the proposed approach}
We now formulate a resilient observer-based H$_\infty$ distributed control protocol for MAS \eqref{eq3}-\eqref{eq4} in the presence of attacks. In the proposed approach, only the leader communicates its actual sensory information and followers do not exchange their actual state information. This stops propagating Type 1 attacks from a compromised agent to others. To this end, the followers estimate the leader's state using a distributed observer and communicate this estimation to their neighbors to achieve consensus on the leader state.

The distributed observer is designed as
\begin{equation}\label{Observ1} 
\left\{ \begin{gathered}
  {{\dot r}_i} = A{r_i} + B{\upsilon _i} \hfill \\
  {\upsilon _i} = cF{\eta _i} + \rho h(F{\eta _i}) \hfill \\ 
\end{gathered} \,\,\forall i=1,\dots,N  \right.
\end{equation}
where $h(.)$ is defined in \eqref{g_Func} and $\eta_i$ is a revisited local neighborhood observer tracking error for agent $i$ defined by 
\begin{equation}\label{ObsrvLNTE}
\eta_i={\sum\limits_{j \in {N_i}} {C_j(t)T_{ij}(t){a_{ij}}({r_j} - {r_i}) + {b_i}({\zeta _0} - {r_i})} }
\end{equation}
where $0\leq C_j(t) \leq 1$ is the confidence of agent $j$ and $0\leq T_{ij}(t) \leq 1$ is the trust value of agent $i$ to its neighbor $j$. The confidence and trust values along with the design parameters $c$, $\rho$ and $F$ in \eqref{Observ1} are designed in subsection B to mitigate Type 2 attacks, i.e., to identify and remove entirely compromised agents or attacks on communication network, and thus guarantee that $r_i\to \zeta_0$ for all intact agents, regardless of attacks. To further increase resiliency at the local level and attenuate the effects of Type 1 attacks on the compromised agent itself, the control input $u_i$ in \eqref{eq3} is designed as a function of $r_i$ and $x_i$ in subsection C (see Theorem 4) to guarantee that the following bounded ${L_2}$-gain condition is satisfied for the agent $i$
\begin{equation}\label{eq23} 
\frac{{\int_0^\infty  {{e^{ - {\alpha _i}t}}{{\left\| {{z_i}(t)} \right\|}^2}dt} }}{{\int_0^\infty  {{e^{ - {\alpha _i}t}}{{\left\| {{\omega _i}(t)} \right\|}^2}dt} }} \leq \gamma _i^2
\end{equation}
with ${z_i}$ defined as the controlled or performance output and is obtained by
\begin{equation}\label{eq22} 
{\left\| {{z_i}} \right\|^2} = {({x_i} - {\zeta_0})^T}{Q_i}({x_i}-{\zeta_0}) + u_i^T{R_i}{u_i}
\end{equation}
where ${\alpha _i}$ and ${\gamma _i}$ represent the discount factor and the attenuation level of the attack $\omega_i$, respectively, and the weight matrices ${Q_i}$ and ${R_i}$ are symmetric positive definite. If condition \eqref{eq23} is satisfied, then, the H${_\infty}$ norm of ${T_{{\omega _i},{z_i}}}$, i.e., the transfer function  from the attack ${\omega _i}$ to the performance output ${z_i}$, is less than or equal to ${\gamma _i}$. Note also that $\omega_i(t)$ does not need to be a bounded energy signal because of the discount factor ${\alpha _i}$. The problem formulation can now be given as follows.  
\smallskip

\noindent
\textbf{Problem 1. (Resilient H${_\infty}$ Synchronization Problem)} Consider $N$ agents defined in \eqref{eq3}-\eqref{eq4} with the distributed observer given by \eqref{Observ1}-\eqref{ObsrvLNTE}. Design the control protocol $u_i=f_i(x_i,r_i,v_i)$ in \eqref{eq3} along with $C_j(t)$, $T_{ij}(t)$, $c$, $\rho$ and $F$ in \eqref{Observ1}-\eqref{ObsrvLNTE} such that  
\begin{enumerate}
\item The bounded ${L_2}$-gain condition \eqref{eq23} is satisfied when ${\omega _i}\ne0$.  
\item The synchronization problem is solved, i.e., $\left\| {{x_i}(t) - {\zeta_0}(t)} \right\| \to 0$, $i = 1, \ldots ,N$ when ${\omega _i} = 0$.  
\end{enumerate}

\subsection{The proposed distributed observer design}
The distributed observer \eqref{Observ1}-\eqref{ObsrvLNTE} only communicates the observer state $r_i$, which cannot be affected by Type 1 attacks on physical components. A trust-confidence mechanism is designed in the following to mitigate Type 2 attacks. To this end, a confidence value is defined for each agent to indicate the trustworthiness of its own observer information. Agents communicate their confidence value with their neighbors to alert them to put less weight on the information they are receiving from them, depending on how low their level of confidence is. This slows down the propagation of Type 2 attacks. If an agent is not confident about its own observer information, it then assigns a trust value to its neighbors and incorporates these trust values along with the confidence values received from neighbors in its update law to determine the significance of the incoming information. Figure \ref{fig:Monitor} shows the block diagram of the proposed distributed monitor.

Note also that it is assumed that the attacker designs its signal based on Theorem 2 to deceive intact agents, so that they cannot monitor any anomaly by examining their local neighborhood tracking error. This is considered the worst attack scenario. If, however, the attacker does not satisfy the conditions of Theorem 2, then, intact agents can easily detect attacks by using Kullback–Leibler divergence criteria to check discrepancy between the normal statistical properties of the local neighborhood tracking error and its actual ones \cite{basseville1993detection}.
\begin{figure}[!ht]
\begin{center}
\includegraphics[width=3in,height=1.8in]{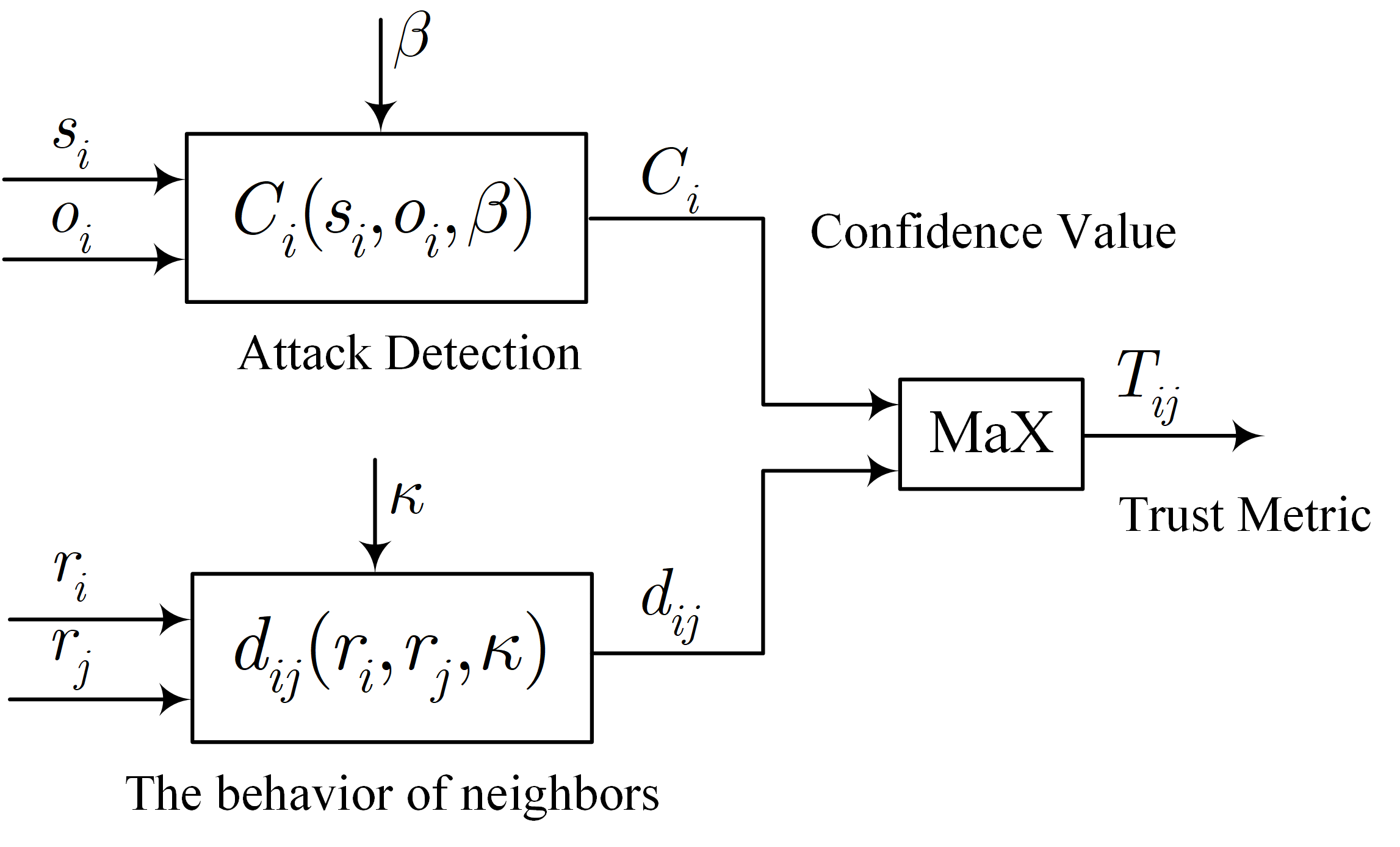}
\vspace{0pt}\caption{The distributed monitor structure.}\label{fig:Monitor}
\captionsetup{justification=centering}
\vspace{-10pt}
\end{center}
\end{figure}

\subsubsection{Confidence value}
A confidence value is defined for each agent which shows the validity of its information. To proceed, define
\begin{equation}\label{eq:o_mon}
{o_i} = \left\| \eta_i \right\|
\end{equation}
for agent $i$ where $\eta_i$ is defined in \eqref{ObsrvLNTE}. Based on Theorem 2, $\eta_i$ and, consequently, $o_i$ converges to zero for intact agents. Now, define
\begin{equation}\label{eq:s_mon}
{s_i} = \sum\limits_{j \in {N_i}} {{a_{ij}}\left\| {{r_j} - {r_i}} \right\| + {b_i}\left\| {{\zeta_0} - {r_i}} \right\|} 
\end{equation}
for agent $i$ . In contrast to $o_i$, $s_i$ does not converge to zero if agent $i$ is in the path of an attacker. This is because $s_i = 0$ requires $r_j=r_i=\zeta_0,\, \forall j \in {N_i}$, which indicates that agent $i$ and its neighbors are synchronized and, therefore, are not in the path of an attacker. In the absence of attack, $o_i$ and $s_i$ converge to zero and have the same behavior. Therefore, by comparing $o_i$ and $s_i$ one can detect whether or not the agent is in the path of a compromised agent. The confidence value $C_i(t)$ in \eqref{ObsrvLNTE} for an intact agent $i$ is defined as
\begin{equation}\label{eq:Confid}
{C_i}(t) = \beta \int_0^t {{e^{-\beta (t-\tau)}}{q_i}(\tau )d\tau }
\end{equation}
with
\begin{equation}\label{eq:Zi}
{q_i}(t) = \frac{{{\Delta _i}}}{{{\Delta _i} + \left\| {{s_i}(t) - {o_i}(t)} \right\|}}
\end{equation}
where $ \beta > 0$ is a discount factor used to determine how much we value the current experience with regard to the past experiences. ${{\Delta _i}}$ is a threshold value to account for factors other than attacks, i.e., channel fading and disturbance. If agent $i$ is not in the path of any compromised agent, $\left\| {{s_i} - {o_i}} \right\|$ is zero almost all the time and, consequently, $C_i$ is almost one. On the other hand, if agent $i$ is affected by an attacker, then $\left\| {{s_i} - {o_i}} \right\| \gg {\Delta _i}$, and $C_i$ is less than one and its value depends on how close the agent is to the source of the attack. Equation \eqref{eq:Confid} can be implemented by the differential equation ${{\dot C}_i}(t) + \beta {C_i}(t) = \beta {q_i}(t)$. The worst case scenario is assumed in which a disrupted agent broadcasts the confidence value 1 to its neighbors to fool them.
\subsubsection{Trust value}
The trust value is defined to determine the importance of the incoming information of each agent's neighbor. To calculate the trust value of agent $i$ to agent $j$ , we first measure the difference between the state of agent $j$ and the average of the state of all neighbors of agent $i$ using
\begin{equation}\label{eq:dij}
{d_{ij}}(t) = \kappa \int_0^t {{e^{-\kappa (t-\tau)}}{l_{ij}}(t)d\tau } 
\end{equation}
with
\begin{equation}\label{eq:lij}
{l_{ij}}(t) = \frac{{{\theta _i}}}{{\left\| {{r_j}(t) - h_i(t)} \right\| + {\theta _i}}},h_i = \frac{1}{{\left| {{N_i}} \right|}}\sum\limits_{j \in {N_i}} {{r_j}}
\end{equation}
where $h_i$ denotes the average value of the neighbors of agent $i$ and $\left| {{N_i}} \right|$ is the number of neighbors of agent $i$ . The discount factor $\kappa > 0$ determines how much we value the current experience of interaction with regard to the past experiences. $\theta_i$ is a threshold value to take into account factors other than attacks. Equation \eqref{eq:dij} can be implemented as ${{\dot d}_{ij}}(t) + \kappa {d_{ij}}(t) = \kappa {l_{ij}}(t)$. Now, we define the trust value of agent $i$ to its neighbor $j$ given as ${T_{ij}}(t)$ in \eqref{ObsrvLNTE} as
\begin{equation}\label{eq:TRust1}
{T_{ij}}(t) = \max \left( {{C_i}(t),{d_{ij}}(t)} \right)
\end{equation}

$T_{ij}$ can also be normalized to satisfy $\sum\nolimits_{j \in {N_i}} {{T_{ij}}(t)}  = 1$. If there is no attack and the network is also synchronized, then, $\left\| {{r_j} - {h_i}} \right\| \, \forall j \in {N_i}$ is zero, and, consequently, $T_{ij}$ is one $\forall j \in {N_i}$. Moreover, when there is no attack and agent $i$ receives considerably different values from its neighbors before synchronization, e.g. as a result of a change in the state of the  leader, since $C_i$ is close to one as there is no attack, $T_{ij}$ is almost one $\forall j \in {N_i}$. On the other hand, if agent $i$ is affected by an attack, then $C_i$ is small and the trust of agent $i$ to agent $j$ depends on $\left\| {{r_j} - {h_i}} \right\|,j \in {N_i}$.

It is shown in Theorem 3 that the proposed observer-based distributed control protocol \eqref{Observ1}-\eqref{ObsrvLNTE} guarantees synchronization of intact agents, if the following assumption is satisfied.
\smallskip

\noindent
\textbf{Assumption 4.} The network connectivity is at least $(2f+1)$, i.e., at least half of the neighbors of each agent are intact \cite{Pasqualetti2013}.
\smallskip

Define
\begin{equation}\label{HLG_eq} 
\mathcal{H}=(\mathcal{L}+G)
\end{equation}
where $\mathcal{L}$ is the graph Laplacian matrix and $G$ is the diagonal pinning matrix.Then, based on Lemma 1, $\mathcal{H}$ is a non-singular M-matrix. The following lemmas are used in the proof of Theorem 3.  
\smallskip

\noindent
\textbf{Lemma 2 \cite{Lewis2011TAC}.} Let $P> 0$ and Assumption 2 be satisfied. Then, the solution $P$ to the following ARE equation 
\begin{equation}\label{ARE_Observ1} 
A^TP+PA+I-PBB^TP=0,
\end{equation}
is positive definite.
\smallskip

\noindent
\textbf{Lemma 3 \cite{Li2015Lemma2}.} Let Assumption 1 be satisfied. Then, there exists a positive vector $\phi=[\phi_1,\dots,\phi_N]^T$ such that $(\Phi \mathcal{H}+\mathcal{H}^T\Phi) >0$ where $\Phi=diag\{\phi_1,\dots,\phi_N\}$, $\phi=(\mathcal{H}^T)^{-1}1_N$, and $\mathcal{H}$ is defined in \eqref{HLG_eq}. 

\begin{theorem}
Let Assumption 4 be satisfied. Consider the dynamic observer defined in \eqref{Observ1}-\eqref{ObsrvLNTE} with $C_j(t)$ given by \eqref{eq:Confid} and $T_{ij}(t)$ defined in \eqref{eq:TRust1}. Let $F= -B^T P$, $c \geqslant {{{\phi _{\max}}}}\big/{{{\lambda _{\min }}(\Phi \mathcal{H} + {\mathcal{H}^T}\Phi)}}$ and $\rho \geq v_{max}$, where $P$ is the solution to the algebraic Riccati equation \eqref{ARE_Observ1} and $\phi _{\max}=\mathop {\max }\limits_{i = 1,\dots,N} ({\phi _i})$. Then, $r_i \to \zeta_0$ for all intact agents.
\end{theorem}
\begin{proof}
See Appendix B.
\end{proof}

\noindent
\textbf{Remark 2.} Note that for Type 1 attacks, the proposed control framework does not impose any constraints on the number of neighbors or the total number of agents under attacks. Note also that in contrast to existing mitigation approaches, we do not discard the neighbor's information for an agent based solely on the difference between their values. Therefore, when the discrepancy between agents is because of a legitimate change in the leader, the confidence and trust values for each agent become 1 and, consequently, all agents synchronize to the leader.

The following subsection shows how to design a resilient observer-based H$_\infty$ distributed controller. Non-homogeneous game AREs are derived for solving the optimal H$_\infty$ synchronization problem.
\subsection{The proposed resilient controller}
It was shown in Theorem 3 that $r_i \to \zeta_0$ for all intact agents regardless of attacks. Similar to \cite{MODARES2016Separ}, one can show that if $u_i(t)$ in \eqref{eq3} is designed to guarantee $x_i \to r_i$, using the separation principle, one can guarantee $x_i \to \zeta_0$. Therefore, in the following, the control input $u_i$ is designed to solve Problem 1 with $\zeta_0$ replaced with $r_i$.

Define the error between the state of agent $i$ and its observer as ${\epsilon_i}(t) = {x_i}(t) - {r_i}(t)$. The dynamic of the error becomes 
\begin{equation}\label{T5_1}
{{\dot \epsilon}_i}(t) = A{\epsilon_i}(t) + B{u_i}(t) + D{\omega _i}(t) - B{\upsilon _i}(t)
\end{equation}

Define the augmented system state as
\begin{equation}\label{eq26} 
{X_i}(t) = [\epsilon_i(t)^T\quad {r_i}{(t)}^T]^T \in {\mathbb{R}^{2n}}
\end{equation}

Using \eqref{T5_1} and \eqref{Observ1} together yields the augmented system as
\begin{equation}\label{eq27} 
{{\dot X}_i} = T{X_i} + {B_1}{u_i} + {D_1}{\omega _i} + {E_1}{\upsilon_i} 
\end{equation}
with
\begin{equation}\label{eq28} 
  T = \left[ {\begin{array}{*{20}{c}}
  A&0 \\ 
  0&A 
\end{array}} \right],{B_1} = \left[ {\begin{array}{*{20}{c}}
  B \\ 
  0 
\end{array}} \right],{D_1} = \left[ {\begin{array}{*{20}{c}}
  D \\ 
  0 
\end{array}} \right] 
  {E_1} = \left[ {\begin{array}{*{20}{c}}
  -B \\ 
  {B} 
\end{array}} \right],
\end{equation}

The performance output \eqref{eq22} in terms of the augmented state \eqref{eq26} (while $\zeta_0$ is replaced with $r_i$) becomes
\[{\left\| {{Z_i}} \right\|^2} = X_i^T{Q_i}{X_i} + u_i^T{R_i}{u_i}, \quad  {Q_i} = \left[ {\begin{array}{*{20}{c}}
  {{Q_{1i}}}&0 \\ 
  0&0 
\end{array}} \right]\]

Using \eqref{Observ1} and \eqref{eq26}, the control protocol for the augmented system can be written as the following non-homogeneous control input
\begin{equation}\label{eq32} 
{u_i} = {K_i}{X_i}+g_i
\end{equation}

Note that $g_i$ is used to compensate the non-homogeneous term ${\upsilon_i}$ in the augmented system \eqref{eq27}. With the aid of \eqref{eq23} (while $\zeta_0$ is replaced with $r_i$), define the discounted performance function in terms of the augmented system \eqref{eq27} as
\begin{equation}\label{eq37}
\begin{gathered}
  J({X_i},{u_i},{\omega _i}) = \hfill \\
  \int_{{\kern 1pt} t}^{{\kern 1pt} \infty } {{e^{ - {\alpha _i}\left( {\tau  - t} \right)}}} \left[ {X_i^T{Q_i}{X_i} + u_i^T{R_i}{u_i} - {\gamma _i}^2{\omega _i}^T{\omega _i}} \right]d\tau 
\end{gathered} 
\end{equation}

The value function for linear systems is quadratic with the form as
\begin{equation}\label{VF_eq} 
{V_i}({X_i}(t)) ={X_i}{(t)^T}{P_i}{X_i}(t)+2X_i^T\Pi_i+\Gamma_i
\end{equation}
and the corresponding Hamiltonian function becomes
\begin{equation}\label{eq39}
H\left( {{X_i},{u_i},{\omega _i}} \right) \triangleq X_i^T{Q_i}{X_i} + u_i^T{R_i}{u_i} - \gamma _i^2\omega _i^T{\omega _i} - {\alpha _i}{V_i}+\frac{dV_i}{dt}
\end{equation}

\noindent
\textbf{Remark 3.} It is assumed here that the full state of agents is available for measurement. However, if not available, the proposed design procedure can be extended for the case of dynamics controllers in which the states of agents are estimated using a local observer. This is because local observers can estimate agents' state without any exchange of information with their neighbors. On the other hand, if the entire agent is compromised and its state observer is manipulated, its neighbors detect it and discard its information using the proposed trust-confidence mechanism.
\begin{theorem}(Non-homogenous game ARE) The optimal solution for the discounted performance function \eqref{eq37} is
\begin{equation}\label{eq41}
  u_i^* =  - {R_i}^{ - 1}B_1^T({P_i}{X_i} + {\Pi _i}) \hfill \\
\end{equation}
where $P_i$ and $\Pi_i$ are the solution of the following non-homogeneous game ARE
\begin{equation}\label{nonhomo_eq41}
\begin{cases}
\begin{gathered}
{P_i}T + {T^T}{P_i} - {\alpha _i}{P_i} - {P_i}{B_1}{R_i}^{ - 1}B_1^T{P_i}+\frac{1}{{\gamma _i^2}}{P_i}{D_1}D_1^T{P_i}+{Q_i} = 0 \hfill \\
\end{gathered} \\
{{\dot \Pi }_i} = \left( {{\alpha _i}I_{2n} + {P_i}{B_1}{R_i}^{ - 1}B_1^T - {T^T} - \frac{1}{{\gamma _i^2}}{P_i}{D_1}D_1^T} \right){\Pi _i} - {P_i}{E_1}{\upsilon _i}\\
{{\dot \Gamma }_i} = {\alpha _i}{\Gamma _i}+{\Pi_i^T}{B_1}{R_i}^{ - 1}B_1^T{\Pi _i}-\frac{1}{{\gamma _i^2}}{\Pi_i^T}{D_1}D_1^T{\Pi _i} - 2\upsilon _i^TE_1^T{\Pi _i}\\
\end{cases}
\end{equation}
\end{theorem}
\begin{proof}
See Appendix C.
\end{proof}

The following theorem shows that the control protocol \eqref{eq41} along with \eqref{nonhomo_eq41} solve the H$_\infty$ synchronization problem.

\begin{theorem}
Consider the MAS \eqref{eq3}-\eqref{eq4} with the observer \eqref{Observ1}-\eqref{ObsrvLNTE}. Let the control input $u_i$ be given as \eqref{eq41}-\eqref{nonhomo_eq41}. Then, Problem 1 is solved. 
\end{theorem}
\begin{proof}
See Appendix D.
\end{proof}

\section{Model-Free Resilient Off-Policy RL for Solving Optimal Synchronization for Intact Agents}
In this section, an RL algorithm is proposed to solve Problem 1 on-line without requiring any knowledge of the agents' dynamics.

The off-policy RL allows separating the behavior policy from the target policy for both control input and attack. In order to find the optimal control \eqref{eq41} without the requirement of the knowledge of the system dynamics, the off-policy RL algorithm \cite{Modares2015Tracking} is used in this subsection. Off-policy algorithm has two separate stages. In the first stage,  an admissible policy is applied to the system and the system information is recorded over the time interval $T$. Then, in the second stage, without requiring any knowledge of the system dynamics, the information gathered in stage 1 is repeatedly used to find a sequence of updated policies $u_i^k$ and  $\omega_i^k$ converging to $u_i^*$ and $\omega_i^*$. To this end, the augmented system dynamics \eqref{eq27} is first written as 
\begin{equation}\label{eq71_1}
{{\dot X}_i} = {T^k}{X_i} + {B_1}\left( {{u_i} - u_i^k} \right) + {D_1}\left( {{\omega_i} - \omega_i^k} \right) + {E_1}{\upsilon_i} 
\end{equation}
where $X_i$ is defined in \eqref{eq26} and $u_i^k$ and $\omega_i^k$ denote the control and disturbance target policies in iteration $k$ to be updated. Now, using \eqref{VF_eq}, one has
\begin{equation}\label{eq71}
{V_i^k}({X_i}(t)) ={X_i}{(t)^T}{P_i^k}{X_i}(t)+2X_i^T\Pi_i^k+\Gamma_i^k
\end{equation}

In this case, the control policy and worst case attack signal \eqref{eq41_1} can be written as
\begin{equation}\label{off_eq41_1}
\begin{gathered}
  u_i^{k+1} =  - {R_i}^{ - 1}B_1^T({P_i^k}{X_i} + {\Pi _i^k}) \hfill \\
  \omega _i^{k+1} = \frac{1}{{\gamma _i^2}}D_1^T({P_i^k}{X_i} + {\Pi _i^k}) \hfill \\ 
\end{gathered}
\end{equation}

Taking time derivative of $V_i^k$ along the augmented system dynamic \eqref{eq71} yields
\begin{equation}\label{off_eq3}
\begin{gathered}
  \dot V_i^k({X_i}(t)) = {\alpha _i}\left( {{X_i}{{(t)}^T}P_i^k{X_i}(t) + 2X_i^T\Pi _i^k + \Gamma _i^k} \right) - X_i^T{Q_i}{X_i} \hfill \\
   - {(u_i^k)^T}{R_i}(u_i^k) - 2{(u_i^{k + 1})^T}{R_i}({u_i} - u_i^k) \hfill \\
   + \gamma _i^2{(\omega _i^k)^T}(\omega _i^k) + 2\gamma _i^2{(\omega _i^{k + 1})^T}({\omega _i} - \omega _i^k) \hfill \\ 
\end{gathered}
\end{equation}

The following off-policy integral RL Bellman equation is derived by multiplying both sides of \eqref{off_eq3} by $e^{-\alpha_i(\tau-t)}$ and integrating 
\begin{equation}\label{eq72}
\begin{gathered}
  {e^{ - {\alpha _i}T}}V_i^k({X_i}(t + T)) - V_i^k({X_i}(t)) \hfill \\
   =  - \int_t^{t + T} {{e^{ - {\alpha _i}(\tau  - t)}}\left( {X_i^T{Q_i}{X_i} + {{(u_i^k)}^T}{R_i}(u_i^k) - \gamma _i^2{{(\omega _i^k)}^T}(\omega _i^k)} \right)} d\tau  \hfill \\
   + \int_t^{t + T} {{e^{ - {\alpha _i}(\tau  - t)}}\left( {2\gamma _i^2{{(\omega _i^{k + 1})}^T}({\omega _i} - \omega _i^k) - 2{{(u_i^{k + 1})}^T}{R_i}({u_i} - u_i^k)} \right)} d\tau  \hfill \\ 
\end{gathered}
\end{equation}

The off-policy RL algorithm presented by iterating on \eqref{eq72} to solve the non-homogeneous game ARE, is listed in Algorithm 1.

The following theorem shows that using the proposed control framework, the learning mechanism is resilient against both Types 1 and 2 attacks.
\begin{theorem}
Consider the MAS \eqref{eq3}-\eqref{eq4} under the control protocol \eqref{eq41} with the observer dynamics given by \eqref{Observ1}-\eqref{ObsrvLNTE}. Let the off-policy Algorithm 1 be used to learn the gains in \eqref{eq41}. Then, Problem 1 is solved, if Assumptions 1 to 4 are satisfied.
\end{theorem}
\begin{proof}
Similar to \cite{Modares2015Tracking}, one can show that the off-policy RL Algorithm 1 solves Problem 1 in an optimal manner, as long as $r_i \to \zeta_0$. That is, for intact agents $x_i \to \zeta_0$ and for a compromised agent the $L_2$ condition \eqref{eq23} is satisfied. This boils down the proof to show that $r_i \to \zeta_0$ even if the system is under attack. On the other hand, Theorem 3 shows that $r_i \to \zeta_0$ regardless of attacks. This completes the proof.
\end{proof}

\noindent
\textbf{Remark 4.} In the proposed Algorithm 1, steps 4 to 9 for finding  the trust and confidence values are continually employed even after learning. However, once the optimal gain is found, the learning steps 10-12 are skipped, because the gains required for the control policy are computed and there is no need for further computation, unless another learning phase is initiated by any change in the agent dynamics. One might argue that the off-policy Algorithm 1 requires to measure the attack signal $\omega_i(t)$ which is restrictive. However, the off-policy Algorithm can learn about the worst-case attack signal using actual measurable disturbances (either applied intentionally or coming from nature), instead of measuring attack signals. Once it learned the worst-case scenario, it can attenuate attacks without measuring them.
\begin{algorithm}[!ht]
\caption{Autonomous Resilient Control Protocol (ARCP) for agent $i$}
\begin{algorithmic}[1]
\State \textbf{Procedure} ARCP
\State \textbf{Apply an admissible behavior policy to the agent}
\State \textbf{For i=1:N}
\State \hspace{0.5cm} \textbf{Use $r_j,\,\, j\in N_i$ to compute $s_i$ and $o_i$}
\State \hspace{0.5cm}\textbf{Compute $C_i$ as} \[{C_i}(t) = \beta \int_0^t {\frac{{{{\Delta _i}e^{ - \beta (t - \tau )}}}}{{{\Delta _i} + \left\| {{s_i}(\tau ) - {o_i}(\tau )} \right\|}}d\tau }\]
\State \hspace{0.5cm}\textbf{Communicate $C_i$ with agents $j\in N_i$}
\State \hspace{0.5cm}\textbf{For each agent $j \in N_i\,$, compute}
\[{d_{ij}}(t) = \kappa \int_0^t {\frac{{{{\theta _i}e^{ - \kappa (t - \tau )}}}}{{{\theta _i} + \left\| {{r_j}(\tau ) - \frac{1}{{\left| {{N_i}} \right|}}\sum\limits_{j \in {N_i}} {{r_j}} } \right\|}}d\tau } \]
\[T_{ij}(t)=max \left(C_i(t),d_{ij}(t)\right)\]
\State \hspace{0.5cm}\textbf{Use $C_j,j\in N_i$ and $T_{ij}$ to find}
\[{\eta_i} = {\sum\limits_{j \in {N_i}} {C_j(t)T_{ij}(t){a_{ij}}({r_j} - {r_i}) + {b_i}({\zeta_0} - {r_i})} }\]
\State \hspace{0.5cm}\textbf{Compute the observer state $r_i$ using update law \eqref{Observ1}}
\State \hspace{0.5cm}\textbf{Collect enough samples $X_i=[X_i;X_i(iT_s)]$ and $u_i$ at different sampling interval $T_s$}
\State \textbf{End for}
\State \textbf{Given $u_i^k$ and collected information solve the following Bellman equation}
\[\begin{gathered}
  {e^{ - {\alpha _i}T}}V_i^k({X_i}(t + T)) - V_i^k({X_i}(t)) =  \hfill \\
   - \int_t^{t + T} {{e^{ - {\alpha _i}(\tau  - t)}}\left( {X_i^T{Q_i}{X_i} + {{(u_i^k)}^T}{R_i}(u_i^k) - \gamma _i^2{{(\omega _i^k)}^T}(\omega _i^k)} \right)} d\tau  \hfill \\
   + 2\gamma _i^2\int_t^{t + T} {{e^{ - {\alpha _i}(\tau  - t)}}{{(\omega _i^{k + 1})}^T}({\omega _i} - \omega _i^k)d\tau }  \hfill \\
   - 2\int_t^{t + T} {{e^{ - {\alpha _i}(\tau  - t)}}{{(u_i^{k + 1})}^T}{R_i}({u_i} - u_i^k)d\tau }  \hfill \\ 
\end{gathered}\]
\State \textbf{Once the learning is done, replace the behavior policy with the control solution found using RL}
\end{algorithmic}
\end{algorithm}
\section{Simulation Results}
In this section, an example is provided to verify the effectiveness of the proposed control protocol. The communication graph is given in Fig. \ref{fig:Exam1}.

\begin{figure}[!ht]
\begin{center}
\includegraphics[width=2.3in,height=1.4in]{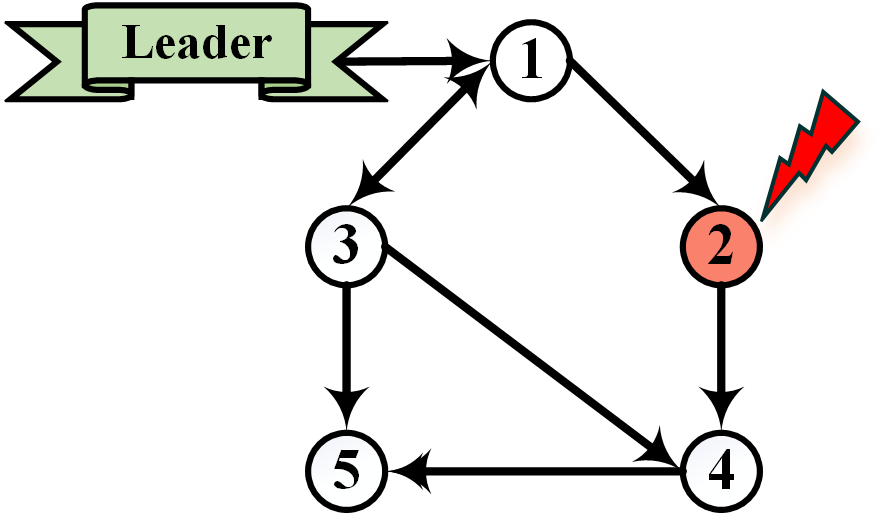}
\vspace{0pt}\caption{The communication graph among the agents.}\label{fig:Exam1}
\captionsetup{justification=centering}
\vspace{-10pt}
\end{center}
\end{figure}

Consider 5 agents with dynamics as
\begin{equation}\label{SR1}
 {{\dot x}_i} = \left[ {\begin{array}{*{20}{c}}  0&-4 \\
  { 1}&0 \end{array}} \right]{x_i} + \left[ {\begin{array}{*{20}{c}}  1 \\   0 
\end{array}} \right]{u_i}+ \left[ {\begin{array}{*{20}{c}}  1 \\   0 
\end{array}} \right]{\omega_i}
\end{equation}

The leader dynamics is given by
\begin{equation}\label{SR2}
{\dot x_0} = \left[ {\begin{array}{*{20}{c}}
  0&{ - 4} \\ 
  1&0 
\end{array}} \right]{x_0} + \left[ {\begin{array}{*{20}{c}}
  1 \\ 
  0 
\end{array}} \right]\left( {4{e^{ - 0.15t}}\sin (2t)} \right)
\end{equation}

The design parameters are $Q_{1i}=100I_2$, $R_i=1$, $\alpha_i=0.1$ and, $\gamma_i = 10$ for all agents. Now, assume that Agent 2 is affected by a Type 1 attack with the attack signal $\omega_2$ given as
\begin{equation}\label{eq78}
{\omega _2} = \left\{ \begin{array}{ll}
         10\sin (2t) & \mbox{$t \geq 10$}\\
         0 & \mbox{$otherwise$}\\
         \end{array} \right.
\end{equation}

\begin{figure}
    \centering
    \begin{subfigure}[b]{0.23\textwidth}
        \centering
    \includegraphics[width=1\linewidth,height=3.5cm]{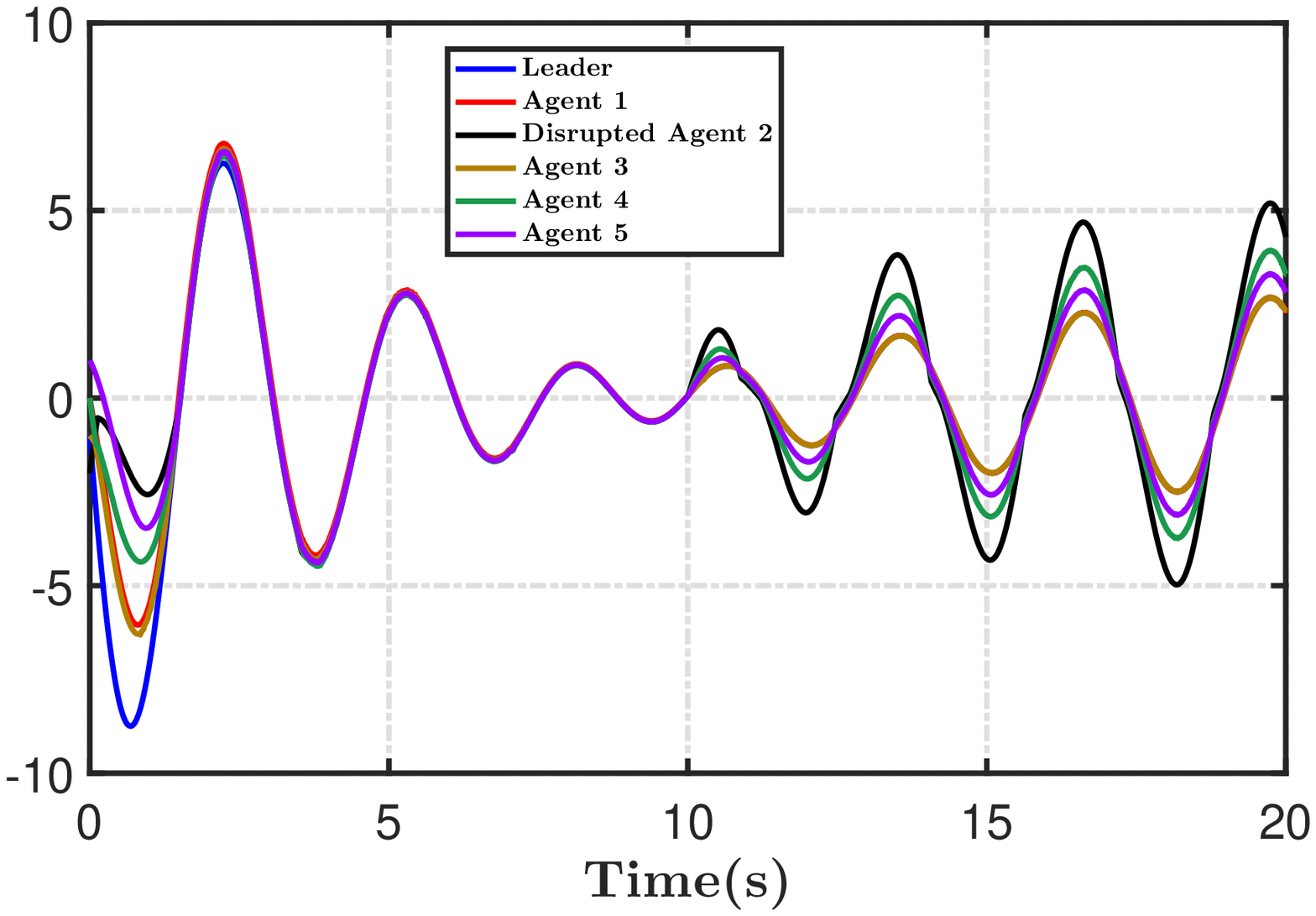}
        \caption{}
        \label{fig:Homo_Std_Attack}
    \end{subfigure}
    \hspace{0cm}
    \begin{subfigure}[b]{0.23\textwidth}
        \centering
        \includegraphics[width=1\linewidth,height=3.5cm]{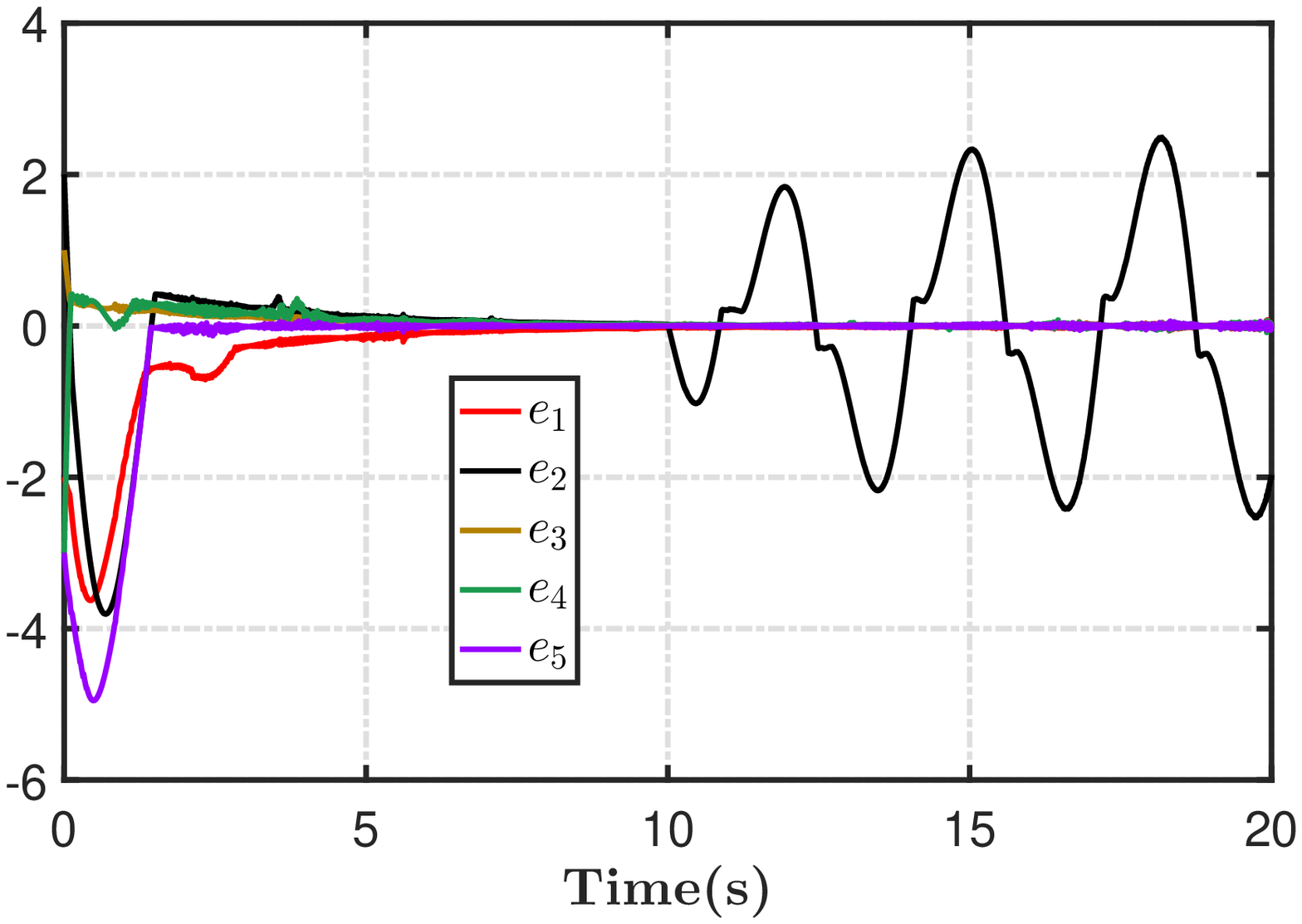}
        \caption{}
        \label{fig:Homo_Observer_Attack}
    \end{subfigure}
    \caption{Agents state, when Agent 2 is under the Type 1 attack \eqref{eq78}. (a) The standard control protocol \eqref{Cnt_S} is used. (b) The local neighborhood tracking error \eqref{eq6} for each agent.}
    \label{fig:Attack effect}
\end{figure}

\begin{figure}[!ht]
\begin{center}
\includegraphics[width=2.5in,height=1.5in]{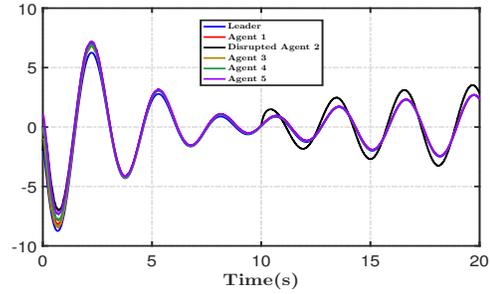}
\vspace{0pt}
\caption{The state of all agents when Agent 2 is under Type 1 attack and the observer defined in \eqref{Observ1} is used for each agent.}\label{fig:Observer_Flayer}
\captionsetup{justification=centering}
\vspace{0pt}
\end{center}
\end{figure}

The agents' states are shown in Fig. \ref{fig:Attack effect}. It is observed from Fig. \ref{fig:Homo_Std_Attack} that when the standard distributed controller \eqref{Cnt_S} is used, before attack all agents synchronize to the leader. However, after the attack, Agents 4 and 5, which have a path to the compromised Agent 2, do not synchronize to the leader. One can see from Fig. \ref{fig:Homo_Observer_Attack} that, as stated in Theorem 2, the local neighborhood tracking error \eqref{eq6} converges to zero for all intact agents except the compromised agent. 

The performance of the observer-based  H$_\infty$ controller \eqref{eq41} in the presence of Type 1 attack \eqref{eq78} is shown in Fig. \ref{fig:Observer_Flayer}. One can see that the compromised agent is the only agent that does not follow the leader. Moreover, the H$_\infty$ controller attenuates the effect of the attack on the disrupted agent, which can be seen by comparing the deviation level of the compromised agent state from its desired value in Figs. \ref{fig:Homo_Std_Attack} and \ref{fig:Observer_Flayer}.  

\begin{figure}
    \centering
    \begin{subfigure}[b]{0.23\textwidth}
        \centering
        \includegraphics[width=1\linewidth,height=3.2cm]{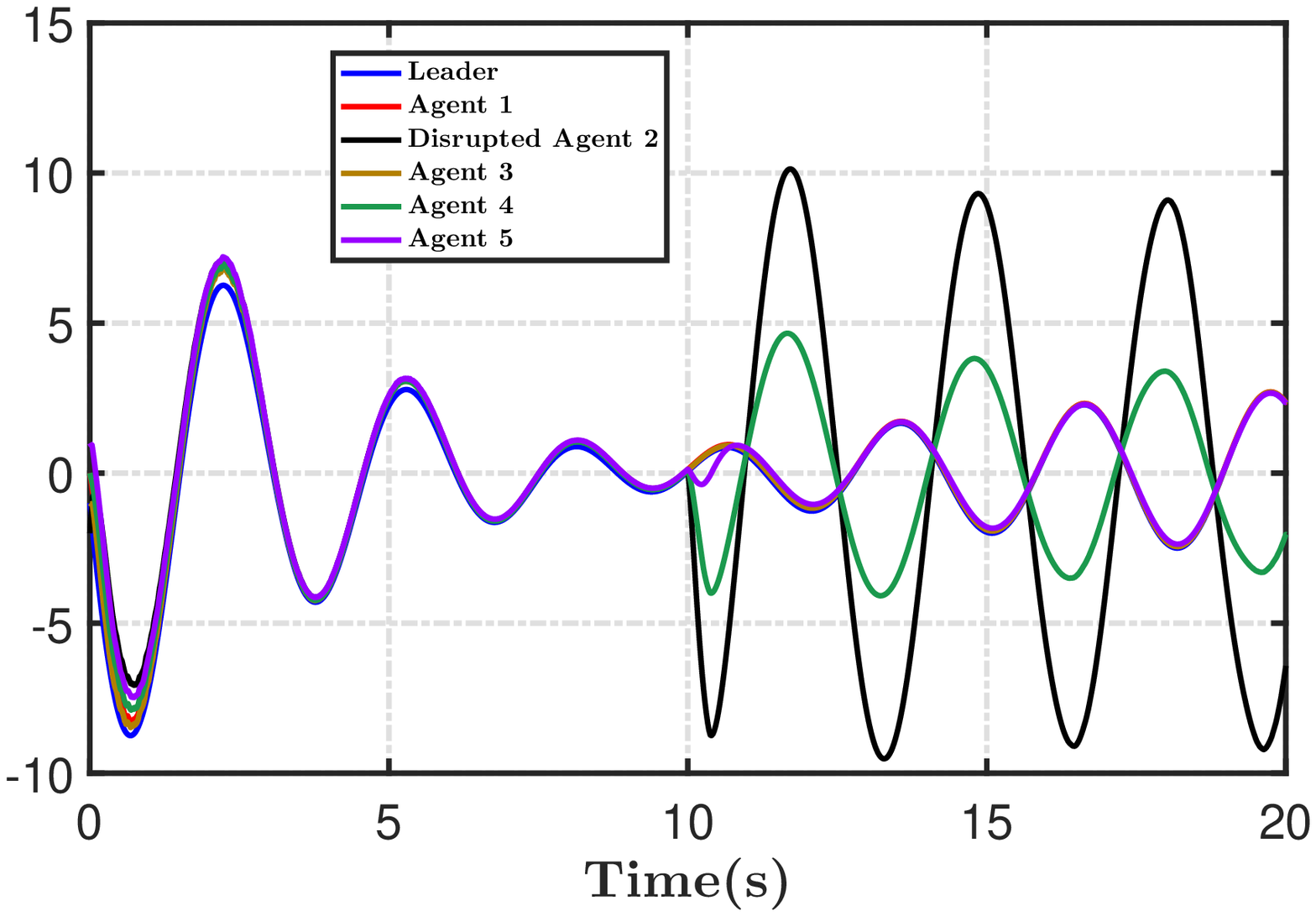}   
        \caption{}
        \label{fig:Obs_Attack}
    \end{subfigure}
   \hspace{0cm}
   \begin{subfigure}[b]{0.23\textwidth}
        \centering
        \includegraphics[width=1\linewidth,height=3.2cm]{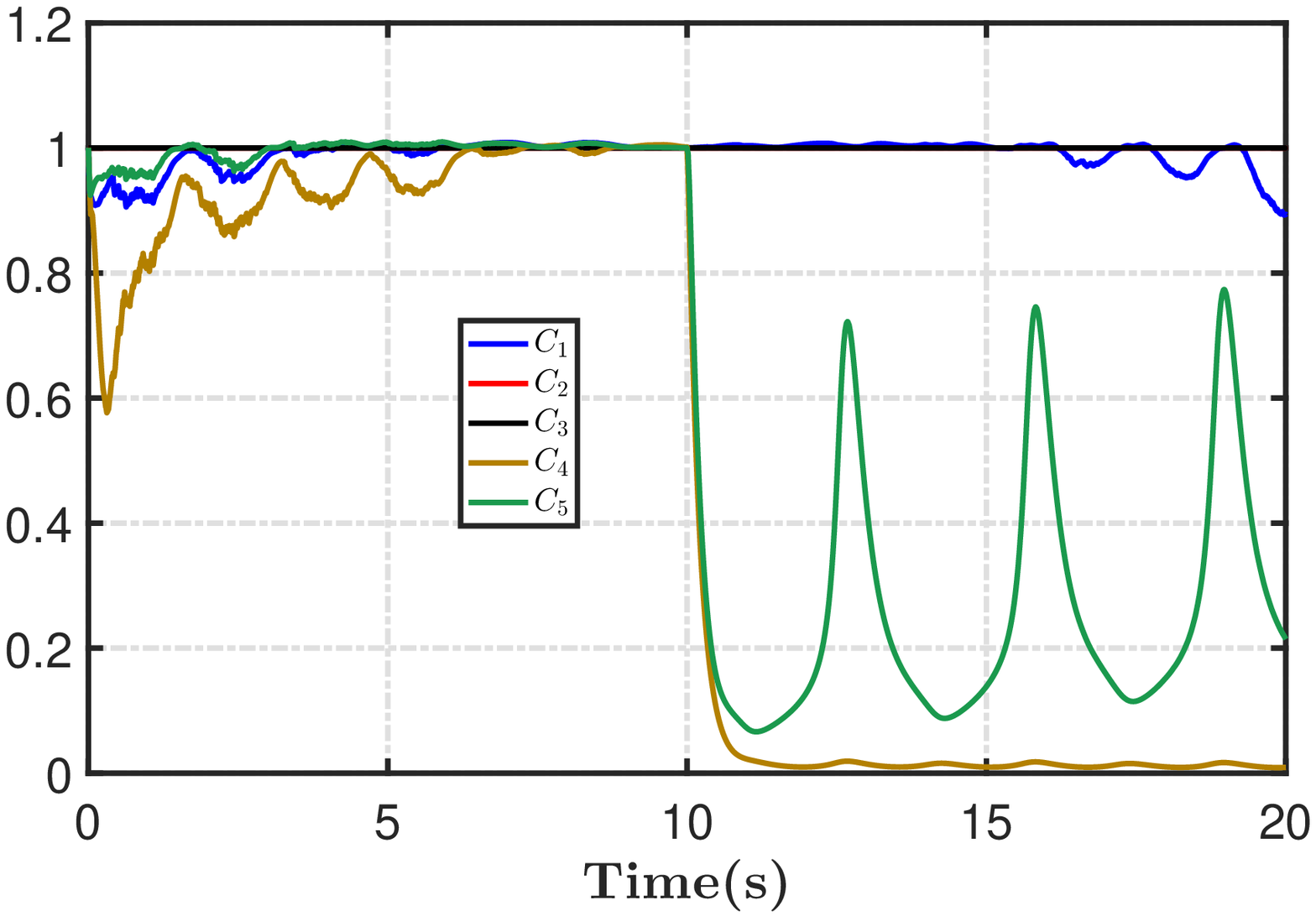}
        \caption{}
        \label{fig:Conf_value}
    \end{subfigure}
    \caption{Agent 2 is under Type 2 attack but Assumption 4 does not hold. (a) Agents state when the control protocol \eqref{eq41} with the observer \eqref{ObsrvLNTE} is used. (b) Confidence values of all agents.}
    \label{fig:TC1}
\end{figure}

Now, assume that Agent 2 is under Type 2 attack. In this case, the attack signal \eqref{eq78} is applied to the observer of Agent 2. However, Assumption 4 is not satisfied. The result is shown in Fig. \ref{fig:TC1}. It can be seen from Fig. \ref{fig:Obs_Attack} that using the trust-confidence mechanism, only the compromised agent and its direct neighbor do not synchronize to the leader. The confidence value of the agents is shown in Fig. \ref{fig:Conf_value}. One can see that Agents 4 and 5 are not confident about their own observer information, since they are in the path of the compromised agent. To satisfy Assumption 4, it is considered that 2 incoming links from Agent 5 and Agent 1 are connected to Agent 4. Fig. \ref{fig:LimEnerg} shows the agents output and confidence value when Assumption 4 is satisfied.  It can be seen from Fig. \ref{fig:TC_New} that only the compromised agent does not synchronize to the leader and all intact agents converge to the leader. The confidence value of the agents is shown in Fig. \ref{fig:CV_New}. One can see that Agent 4 is not confident about its own observer information since it is the only immediate neighbor of the compromised agent.
 
\begin{figure}
    \centering
    \begin{subfigure}[b]{0.23\textwidth}
        \centering
        \includegraphics[width=1\linewidth,height=3.2cm]{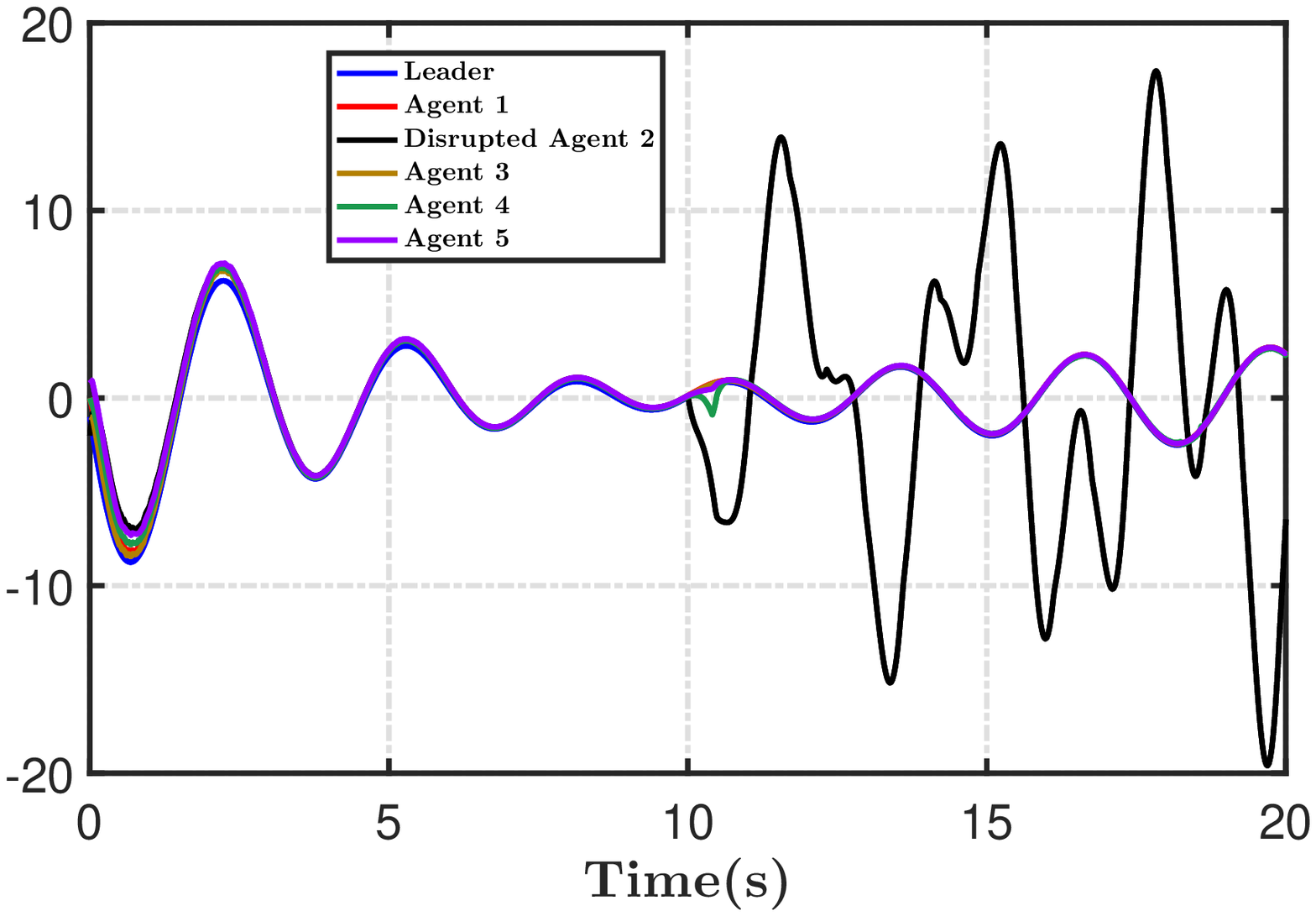}   
        \caption{}
        \label{fig:TC_New}
    \end{subfigure}
    \begin{subfigure}[b]{0.23\textwidth}
        \centering
        \includegraphics[width=1\linewidth,height=3.2cm]{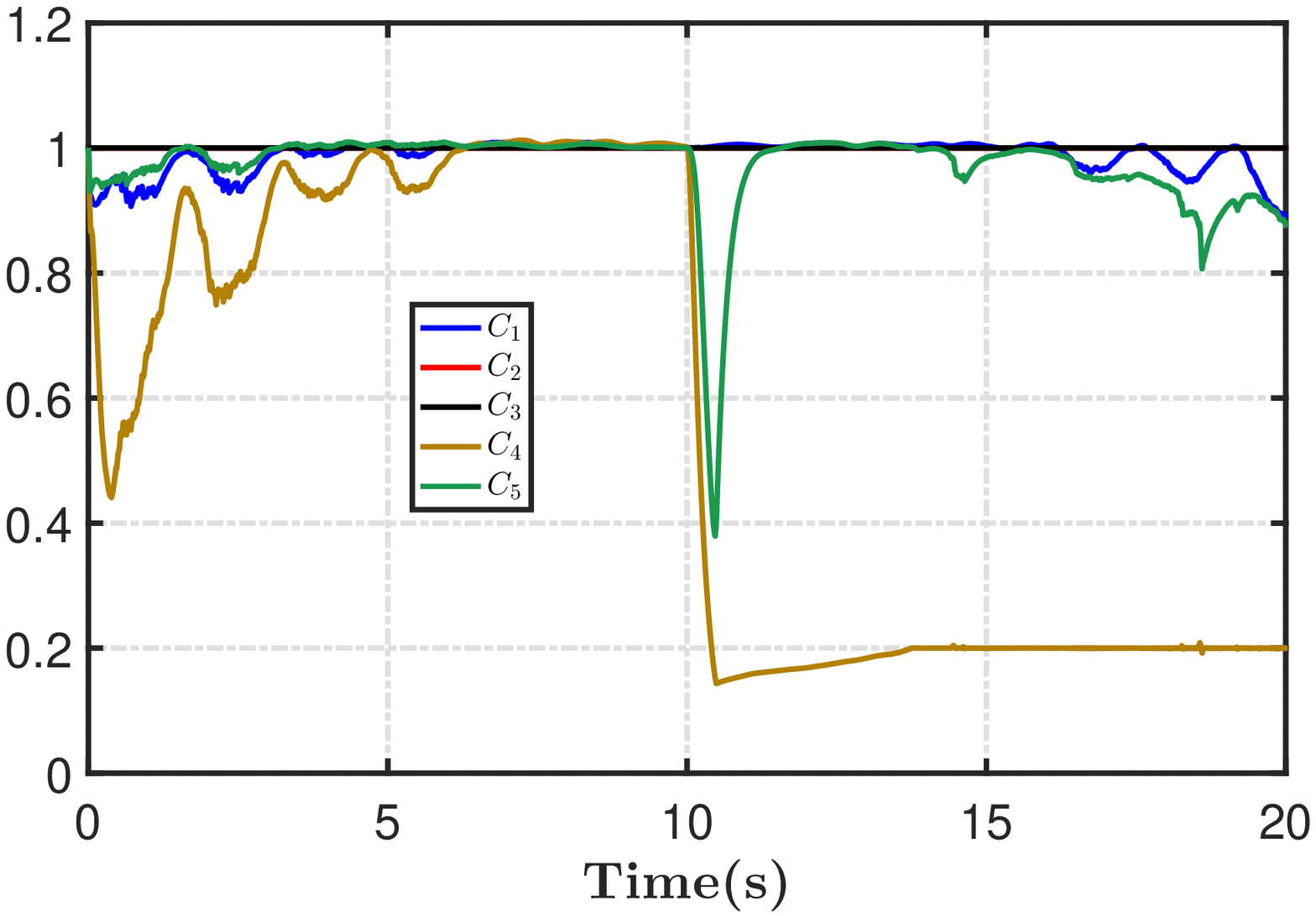}
        \caption{}
        \label{fig:CV_New}
    \end{subfigure}
    \caption{Agent 2 is under Type 2 attack and Assumption 4 holds. (a) Agents state when the control protocol \eqref{eq41} with the observer \eqref{ObsrvLNTE} is used. (b) Confidence values of all agents.}
    \label{fig:LimEnerg}
\end{figure}

\section{Conclusion}
A resilient autonomous control framework is proposed for a leader-follower MAS with an active leader. It is first shown that existing standard synchronization control protocols are prone to attacks. Then, a resilient learning-based control protocol is presented to find optimal solutions to the synchronization problem in the presence of attacks and system dynamic uncertainties. A distributed observer-based H${_\infty}$ controller is first designed to prevent propagating the effects of attacks on sensors and actuators throughout the network, as well as attenuating the effect of these attacks on the compromised agent itself. Non-homogeneous game algebraic Riccati equations are derived to solve the H$\infty$ optimal synchronization problem. Off-policy reinforcement learning is utilized to learn their solution without requiring any knowledge of the agent's dynamics. Then, a trust-confidence based distributed control protocol is proposed to mitigate attacks that hijack the entire node and attacks on communication links. It is shown that the proposed RL-based H${_\infty}$ control protocol is resilient against attacks.

\appendices 
\section{Proof of Theorem 2}

Let $\bar{\mathcal{L}}$ be the graph Laplacian matrix of the entire network, in which the leader is the only root node and $N$ followers are non-root nodes. Then, it can be partitioned as
\begin{equation}\label{eq:Lp}
\bar {\mathcal{L}} = 
\left[ 
\begin{array}{@{}c|c@{}}
   \begin{array}{@{}cccc@{}}
      0 \\
   \end{array} 
      & 0 \dots 0 \\
   \cmidrule[0.4pt]{1-2}
   -\Delta & \mathcal{L}_f\\
\end{array}
\right]\in {\mathbb{R}^{(N+1) \times (N+1)}}
\end{equation}
where $\Delta\in {\mathbb{R}^{N \times 1}}$ denotes a vector whose $i$-th element is nonzero and indicates that the follower $i$ is connected to the leader. $\mathcal{L}_f\in {\mathbb{R}^{N \times N}}$ indicates the interaction between the leader and the followers. Without loss of generality, assume $D=B$. Using \eqref{g_Func}, the global dynamic of the MAS \eqref{eq3} under attack with the control input \eqref{Cnt_S} in terms of the Laplacian matrix \eqref{eq:Lp} and after some manipulations, one has
\begin{equation}\label{M_sen1new1} 
\begin{gathered}
\dot x = \hfill \\
\begin{cases}
    \begin{gathered}
  \left( {{I_{N + 1}} \otimes A} \right)x + \left[ {\begin{array}{*{20}{c}}
  B \\ 
  {{0}} 
\end{array}} \right]{v_0}+ \left( {{I_{N + 1}} \otimes B} \right) \hfill \\
   \left( { - ({c_1} + {{{c_2}}}\Big/{\left\| {\left( {(\bar{\mathcal{L}} \otimes K} \right)x } \right\|})\bar{\mathcal{L}}  \otimes K + \bar \omega } \right) \hfill \\ 
\end{gathered}, &if {\left\| {\left( {(\bar{\mathcal{L}} \otimes K} \right)x } \right\| \neq 0}\\
   \begin{gathered}
  \left( {{I_{N + 1}} \otimes A} \right)x + \left[ {\begin{array}{*{20}{c}}
  B \\ 
  {{0}} 
\end{array}} \right]{v_0} \hfill \\
   + \left( {{I_{N + 1}} \otimes B} \right)\left( { - {c_1}\bar{\mathcal{L}}  \otimes K + \bar \omega } \right) \hfill \\ 
\end{gathered} ,& if {\left\| {\left( {(\bar{\mathcal{L}} \otimes K} \right)x } \right\| = 0}
\end{cases}
\end{gathered}
\end{equation}
where ${x}=[\zeta_0^T,x_1^T,\dots,x_N^T]^T$ and $\bar{\omega}=[0,\omega_1^T,\dots,\omega_N^T]^T$, and $0$ in $\bar{\omega}$ indicates that the leader is a trusted node and is not under attack. It can be seen that agents reach a steady state, i.e., $\dot x_i \to Ax_i+[B^T 0]^T v_0$, if the last terms of \eqref{M_sen1new1} tend to zero, i.e., $\bar{\omega}  \in \operatorname{Im} (\bar c(\bar{\mathcal{L}}  \otimes K))$, where $\bar {c} = ({c_1} + {{{c_2}}}/{\| {{(\bar{\mathcal{L}} \otimes K)}x } \|})$ or $\bar {c} = {c_1}$. Otherwise, since the attack signal has common eigenvalues with the agent dynamics, the agents' states go to infinity. In the latter case, the local neighborhood tracking error goes to zero as $x_i \to \infty$. To prove the former case, we first show that $\bar{\omega} \in \operatorname{Im} (\bar{c}(\bar{\mathcal{L}}  \otimes K))$, if the attack signal is designed as given in the statement of Theorem. Note that $\bar{\omega} \in \operatorname{Im} (\bar{c}(\bar{\mathcal{L}}  \otimes K))$, if there exists a nonzero vector $x_s$ such that
\begin{equation}\label{M_sen3} 
\bar c(\bar{\mathcal{L}}  \otimes K)x_{ss}=\bar{\omega}
\end{equation}

Define $x_{ss}=[\zeta_{s}^T,x_{s}^T]^T$, where $\zeta_{s}$ and $x_{s}$ denote the steady states of the leader and followers, respectively. Since the leader is not under attack, using \eqref{eq:Lp} and \eqref{M_sen3}, one has $\dot \zeta_{s} = A\zeta_0+Bv_0$ and for the followers,
\begin{equation}\label{M_sen44} 
  -\bar c(\Delta \otimes K) \zeta_{s} + \bar c({\mathcal{L}_f} \otimes K){x_{s}} = \omega
\end{equation}

Based on Assumption 1, the followers have at least one incoming link from the leader. On the other hand, ${\mathcal{L}_f}$ captures the interaction among all followers, as well as the incoming link from the leader. The former is a positive semi-definite Laplacian matrix and the latter is a diagonal matrix with at least one nonzero positive element added to it. Therefore, as stated in Lemma 1, ${\mathcal{L}_f}$ is nonsingular and, thus, the solution to \eqref{M_sen44} becomes
\begin{equation}\label{M_sen44S} 
  {x_{s}} = (\bar c({\mathcal{L}_f} \otimes K))^{-1}(\omega + \bar c(\Delta \otimes K) \zeta_{s})
\end{equation}

Since the eigenvalues of the attack signal are assumed a subset of the eigenvalues of the agent dynamics $A$, for every $\bar{\omega}$ there exists a nonzero vector $x_{ss}$ such that \eqref{M_sen3} holds. Therefore, $\bar{\omega} \in \operatorname{Im} (\bar c(\bar{\mathcal{L}}  \otimes K))$. Now, using \eqref{eq:Lp}, the global form of the state neighborhood tracking error \eqref{eq6} can be written as  
\begin{equation}\label{M_sen2}
e = -(\bar{\mathcal{L}} \otimes I_n)x
\end{equation}

Since \eqref{M_sen3} is satisfied, one has
\begin{equation}\label{M_sen3ss} 
\bar c(\bar{\mathcal{L}}  \otimes K)x_s =\bar{\omega} \Rightarrow \bar c(I_{N} \otimes K)e=-\bar{\omega} 
\end{equation}
or, equivalently
\begin{equation}\label{M_sen3ss1} 
\bar cKe_i=-{\omega_i} 
\end{equation}

Equation \eqref{M_sen3ss1} implies that the local neighborhood tracking error is zero, i.e., $e_j=0$ for intact agents that are not directly under attacks, i.e., ${\omega_j}=0 \,\forall j\ne i$. This completes the proof.

\section{Proof of Theorem 3}
Type 1 attacks cannot affect the observer state since the observer cannot be physically affected by an attacker. On the other hand for Type 2 attacks, based on Assumption 4, the total number of compromised agents is assumed less than half of the network connectivity, i.e., $2f+1$. Therefore, even if $f$ neighbors of an intact agent are attacked and collude to send the same value to misguide it, there still exists $f+1$ intact neighbors that communicate values different than the compromised ones. Thus, $r_j \ne r_i$ for some $j \in {N_i}$ and, therefore, although $o_i$ in \eqref{eq:o_mon} is zero based on Theorem 2, $s_i$ in \eqref{eq:s_mon} is nonzero and, consequently, its confidence value $C_i$ in \eqref{eq:Confid} will decrease and the attack will be detected. Moreover, since at least half of its neighbors are intact, it can update its trust values to remove the compromised neighbors. On the other hand, the entire network is still connected to the agent under attack and, therefore, the graph is still connected with the intact agents. Therefore, there exists a spanning tree in the graph associated with all intact agents. Let ${\mathcal{L}_I}(t)$ be the graph of remaining intact agents with $\bar{a_{ij}}=C_j(t)T_{ij}(t)a_{ij}$ as its weights. Since ${\mathcal{L}_I}(t)$ has a spanning tree, based on the above discussion, the inequality defined in Lemma 2 is still satisfied.

The global form of the observer \eqref{Observ1} can be written as
\begin{equation}\label{T3_GOV}
\dot{r}=(I_N\otimes A)r+(c(\mathcal{L}_I(t)+G) \otimes BF)\eta+(\rho I_N\otimes B)\Psi
\end{equation}
where $r=[r_1^T,\dots,r_N^T]^T$ denotes the global state vector of the agents observer and $\Psi=[h(F\eta_1),\dots,h(F\eta_N)]^T$. The global form of \eqref{ObsrvLNTE} is
\begin{equation}\label{T3_GNTE}
\eta = ((\mathcal{L}_I(t)+G)\otimes I_n)(r-{{{\underline{\zeta}}_0}})
\end{equation}
where $\eta = [\eta_1^T,\dots,\eta_N^T]^T$ and ${{{\underline{\zeta}}_0}}= (1_N\otimes I_n)\zeta_0$. Using \eqref{T3_GOV} and defining $\mathcal{H}_I(t)=\mathcal{L}_I(t)+G$, the dynamic of the tracking error \eqref{T3_GNTE} becomes
\begin{equation}\label{T3_GDTE}
\dot\eta = (I_N \otimes A)\eta+(c\mathcal{H}_I(t)\otimes BF)\eta+(\mathcal{H}_I(t)\otimes B)(\rho \Psi-\underline{v}_0)
\end{equation}
with $\underline{v}_0 = 1_N \otimes v_0$. Consider the Lyapunov candidate function as
\begin{equation}\label{T3_LF}
V(\eta) = \eta^T(\Phi \otimes P) \eta
\end{equation}
where $\Phi$ is defined in Lemma 3 and is a positive definite matrix. Using \eqref{T3_GDTE}, the time derivative of \eqref{T3_LF} yields
\begin{equation}\label{T3_DLF}
\begin{gathered}
  \dot V = \underbrace {2{\eta ^T}(\Phi  \otimes PA)\eta  + 2{\eta ^T}(c\Phi \mathcal{H}_I(t) \otimes PBF)\eta }_{{V_1}} \hfill \\
    + \underbrace{2{\eta ^T}(\Phi \mathcal{H}_I(t) \otimes PB)(\rho \Psi - \underline{v}_0)}_{{V_2}}
\end{gathered}
\end{equation}

Substituting $F= -B^T P$, $V_1$ can be expressed as 
\begin{equation}\label{T3_V1}
V_1 = \eta^T(\Phi \otimes (PA+A^TP))\eta-c\eta^T((\Phi\mathcal{H}_I(t)+\mathcal{H}_I^T(t)\Phi) \otimes PBB^TP)\eta
\end{equation}

Based on Lemma 3, $(\Phi\mathcal{H}_I(t)+\mathcal{H}_I^T(t)\Phi)>0$ and the inequality $- c(\Phi \mathcal{H}_I(t) + {\mathcal{H}_I^T(t) \Phi}) \leq  - c{\lambda _{min}}(\Phi \mathcal{H}_I(t) + {\mathcal{H}_I^T(t)}\Phi) \leq  - c{{{\lambda _{\min }}(\Phi \mathcal{H}_I(t) + {\mathcal{H}_I^T(t)} \Phi)}}\Phi\Big/{{{\phi _{\max }}}}$ holds \cite{Wan2017}. Using the Kronecker property $A\otimes (B+C)=A\otimes B+A\otimes C$ and ARE \eqref{ARE_Observ1}, \eqref{T3_V1} becomes
\begin{equation}\label{T3_V12}
\begin{gathered}
  {V_1} =\hfill \\ {\eta ^T}(\Phi  \otimes (PA + {A^T}P))\eta  - c{\eta ^T}((\Phi \mathcal{H}_I(t)+ {\mathcal{H}_I^T(t)}\Phi ) \otimes PB{B^T}P)\eta  \hfill \\
   \leq  - {\eta ^T}(\Phi  \otimes I_n)\eta  \hfill \\  +\left(1 - c\frac{{{\lambda _{\min }}(\Phi \mathcal{H}_I(t) + {\mathcal{H}_I^T(t)} \Phi)}}{{{\phi _{\max }}}} \right){\eta ^T}(\Phi  \otimes PB{B^T}P)\eta  \hfill \\ 
\end{gathered}
\end{equation}

Therefore, $V_1 \leq 0$ if $c \geq {{{\phi _{\max }}}}\big/{{{\lambda _{\min }}(\Phi \mathcal{H}_I(t) + {\mathcal{H}_I^T(t)}\Phi)}}$, which is satisfied if the condition of $c$ in the statement of the theorem is held. Now, for $V_2$ one has
\begin{equation}\label{T3_V2}
{V_2} = \sum\limits_{i = 1}^N {\sum\limits_{j = 1}^N {{\alpha _{ij}}\eta _j^TPB} } (\rho {\psi _i} - {v_0}) \leq -(\rho  - {v_M})\sum\limits_{i = 1}^N {\left\| {B^TP\sum\limits_{j = 1}^N {{\alpha _{ij}}\eta _j^T} } \right\|}
\end{equation}

This implies that $V_2 \leq 0$, if $\rho \geq v_M$. Therefore, $\dot V=V_1+V_2 \leq 0$, if the conditions defined in the statement of the theorem are satisfied. Since $\dot V(t) \leq 0$, $V(t)$ is bounded. Note that $\dot V(t)\equiv 0$ implies that $\eta=0$. Hence, using LaSalle’s invariance principle \cite{isidori1995nonlinear}, it follows that the observer state $r_i$ asymptotically converges to the leader state $\zeta_0$.  

\section{Proof of Theorem 4}
Using the value function \eqref{VF_eq} for the left-hand side of \eqref{eq37} and differentiating along with the augmented system \eqref{eq27} gives the following Bellman equation
\begin{equation}\label{eq39_1}
\begin{gathered}
  H\left( {{X_i},{u_i},{\omega _i}} \right) = 2X_i^T{P_i}(T{X_i} + {B_1}{u_i} + {D_1}{\omega _i} + {E_1}{\upsilon _i}) \hfill \\
   + 2{(T{X_i} + {B_1}{u_i} + {D_1}{\omega _i} + {E_1}{\upsilon _i})^T}{\Pi _i} + 2X_i^T{{\dot \Pi }_i} + {{\dot \Gamma }_i} \hfill \\
   - {\alpha _i}\left( {{X_i}{{(t)}^T}{P_i}{X_i}(t) + 2X_i^T{\Pi _i} + {\Gamma _i}} \right) \hfill \\
   + X_i^T{Q_i}{X_i} + u_i^T{R_i}{u_i} - \gamma _i^2\omega _i^T{\omega _i} = 0 \hfill \\ 
\end{gathered}
\end{equation}

By applying the stationary conditions \cite{lewis2012optimal} as ${{\partial {H_i}} \mathord{\left/
 {\vphantom {{\partial {H_i}} {\partial {u_i^*}}}} \right.
 \kern-\nulldelimiterspace} {\partial {u_i}}} = 0$, ${{\partial {H_i}} \mathord{\left/
 {\vphantom {{\partial {H_i}} {\partial {\omega _i^*}}}} \right.
 \kern-\nulldelimiterspace} {\partial {\omega _i}}} = 0$, the optimal control and the worst case attack signal can be written as
\begin{equation}\label{eq41_1}
\begin{gathered}
  u_i^* =  - {R_i}^{ - 1}B_1^T({P_i}{X_i} + {\Pi _i}) \hfill \\
  \omega _i^* = \frac{1}{{\gamma _i^2}}D_1^T({P_i}{X_i} + {\Pi _i}) \hfill \\ 
\end{gathered}
\end{equation}

Substituting \eqref{eq41_1} into Bellman equation \eqref{eq39_1} results in the following tracking game ARE
\begin{equation}\label{eq42}
\begin{gathered}
H\left( {{X_i},u_i^*,\omega _i^*} \right) = X_i^T{\left( {{P_i}T + {T^T}{P_i} - {\alpha _i}{P_i} - {P_i}{B_1}{R_i}^{ - 1}B_1^TP} \right._i} \hfill \\
\left. { + \frac{1}{{\gamma _i^2}}{P_i}{D_1}D_1^T{P_i} + {Q_i}} \right){X_i} + 2X_i^T\left( {{{\dot \Pi }_i} - {P_i}{B_1}{R_i}^{ - 1}B_1^T{\Pi _i} + } \right. \hfill \\
+ \left. {\frac{1}{{\gamma _i^2}}{P_i}{D_1}D_1^T{\Pi _i} + {P_i}{E_1}{\upsilon _i} - {\alpha _i}{\Pi _i} + {T^T}{\Pi _i}} \right) \hfill \\
+ \left( {{{\dot \Gamma }_i} - {\alpha _i}{\Gamma _i} - \Pi _i^T{B_1}R_i^{ - 1}B_1^T{\Pi _i} + \frac{1}{{\gamma _i^2}}\Pi _i^T{D_1}D_1^T{\Pi _i} + 2\upsilon _i^TE_1^T{\Pi _i}} \right)=0 \hfill \\ 
\end{gathered}
\end{equation}

Since \eqref{eq42} is satisfied for all $X_i$, this can occur if, and only if, \eqref{nonhomo_eq41} holds. This completes the proof.

\section{Proof of Theorem 5}
The Hamiltonian function \eqref{eq39} for the optimal value function $V_i^*$, and any control policy $u_i$ and disturbance policy $\omega_i$ become
\begin{equation}\label{T6_1}
\begin{gathered}
  H(V_i^*,{u_i},{d_i}) = X_i^T{Q_i}{X_i} + u_i^T{R_i}{u_i} - \gamma _i^2\omega _i^T{\omega _i} - {\alpha _i}V_i^* \hfill \\
   + V_{iX}^{*T}\left( {T{X_i} + {B_1}{u_i} + {D_1}{\omega _i} + {E_1}{\upsilon _i}} \right) \hfill \\ 
\end{gathered}
\end{equation}

Using \eqref{eq42}, one has $H(V_i^*,{u_i},{d_i}) = H(V_i^*,u_i^*,\omega _i^*) + {({u_i} - u_i^*)^T}{R_i}({u_i} - u_i^*)+{\gamma ^2}{({\omega _i} - \omega _i^*)^T}({\omega _i} - \omega _i^*)$ and based on Hamiltonian equation \eqref{eq39}, $ H(V_i^*,u_i^*,\omega _i^*)=0$. Then \eqref{T6_1} gives $ X_i^T{Q_i}{X_i} + u_i^T{R_i}{u_i} - \gamma _i^2\omega _i^T{\omega _i} - {\alpha _i}V_i^* + V_{iX}^{*T}\left( {T{X_i} + {B_1}{u_i} + {D_1}{\omega _i} + {E_1}{\upsilon _i}} \right) =  - {({u_i} - u_i^*)^T}{R_i}({u_i} - u_i^*) - {\gamma ^2}{({\omega _i} - \omega _i^*)^T}({\omega _i} - \omega _i^*)$. Now, using the optimal control policy $u_i=u_i^*$ yields
\begin{equation}\label{T6_4}
\begin{gathered}
  X_i^T{Q_i}{X_i} + u_i^T{R_i}{u_i} - \gamma _i^2\omega _i^T{\omega _i} - {\alpha _i}V_i^* \hfill \\
   + V_{iX}^{*T}\left( {T{X_i} + {B_1}{u_i} + {D_1}{\omega _i} + {E_1}{\upsilon _i}} \right) \hfill \\ 
   =  - {\gamma ^2}{({\omega _i} - \omega _i^*)^T}({\omega _i} - \omega _i^*) \leq 0
\end{gathered}
\end{equation}

Multiplying both sides of \eqref{T6_4} by $e^{-\alpha_ it}$, and defining $\dot V_i^* = V_{iX}^{*T}\left( {T{X_i} + {B_1}{u_i} + {D_1}{\omega _i} + {E_1}{\upsilon _i}} \right)$ as the derivative of $V_i^*$ along the trajectories of the closed-loop system, gives \begin{equation}\label{T6_5}
{d}/{{dt}}\left( {{e^{ - {\alpha _i}t}}V_i^*} \right) \leq {e^{ - {\alpha _i}t}}\left( { - X_i^T{Q_i}{X_i} - u_i^T{R_i}{u_i} + \gamma _i^2\omega _i^T{\omega _i}} \right)
\end{equation}

Integrating both sides of \eqref{T6_5} and using the fact that $V_i^*(.)\geq 0$, for every $T>0$ and every $\omega_i \in L_2[0,\infty)$, one has 
\begin{equation}\label{T6_61}
\begin{gathered}
  \int_0^T {{e^{ - {\alpha _i}\tau }}\left( {X_i^T{Q_i}{X_i} + u_i^{*T}{R_i}u_i^*} \right)d\tau }  \hfill \\
   \leq \int_0^T {{e^{ - {\alpha _i}\tau }}\gamma _i^2\omega _i^T{\omega _i}d\tau }  + V_i^*(X(0)) \hfill \\ 
\end{gathered}
\end{equation}

Using the separation principle for the combination of the observer and the controller completes the proof of Part 1 of Problem 1. The proof of Part 2 is similar to \cite{Modares2015Tracking} and, therefore, is omitted.
\bibliographystyle{ieeetr}

\bibliography{reference}

\addtolength{\textheight}{-3cm}

\end{document}